\documentclass[conference]{IEEEtran}
\usepackage[utf8]{inputenc}

\usepackage{hyperref}
\usepackage{amsthm}
\usepackage{amsmath}
\usepackage{amssymb}
\usepackage{latexsym}
\usepackage{cmll}
\usepackage{stmaryrd}
\usepackage{tikz}
\usepackage{tikz-cd}
\usepackage{rotating}
\usetikzlibrary{matrix,shapes}
\usepackage{pgfplots}
\usepackage{wasysym,stackengine,makebox}
\usepackage{enumitem}

\usepackage{ebproof}

\usepackage[switch]{lineno}
\usepackage[colorinlistoftodos,prependcaption,disable]{todonotes}

\newcommand\CMLLPAR{
\usepackage{cmll}
\newcommand\IPar{\mathord{\parr}}
}

\CMLLPAR

\newlist{Itemize}{itemize}{1}
\setlist[Itemize,1]{label=\textbullet,leftmargin=0.8em}

\newlist{Enumerate}{enumerate}{1}
\setlist[Enumerate]{label*=\arabic*,leftmargin=0.8em}

\theoremstyle{plain}
\newtheorem{theorem}{Theorem}
\newtheorem{lemma}[theorem]{Lemma}

\theoremstyle{definition}
\newtheorem{definition}[theorem]{Definition}
\newtheorem{proposition}[theorem]{Proposition}
\newtheorem{example}[theorem]{Example}
\theoremstyle{remark}
\newtheorem{remark}{Remark}

\renewenvironment{proof}{\begin{IEEEproof}}{\end{IEEEproof}}

\newcommand{\Endproof}{
  \ifmmode 
  \else \leavevmode\unskip\penalty9999 \hbox{}\nobreak\hfill
  \fi
  \quad\hbox{$\Box$}
  \par\medskip}

\newcommand\Eqref[1]{(\ref{#1})}

\renewcommand{\phi}{\varphi}
\renewcommand\epsilon{\varepsilon}

\newcommand\Equiv{\Leftrightarrow}
\newcommand{\St}{\mid}

\newcommand{\arrow}{\rightarrow}

\newcommand{\Fbot}{{\mathord{\perp}}}
\newcommand{\Top}{\top}

\newcommand\cA{\mathcal{A}}
\newcommand\cB{\mathcal{B}}
\newcommand\cC{\mathcal{C}}
\newcommand\cD{\mathcal{D}}

\newcommand\cF{\mathcal{F}}

\newcommand\cI{\mathcal{I}}

\newcommand\cK{\mathcal{K}}
\newcommand\cL{\mathcal{L}}
\newcommand\cM{\mathcal{M}}

\newcommand\cS{\mathcal{S}}
\newcommand\cT{\mathcal{T}}

\newcommand\cX{\mathcal{X}}

\newcommand\Fini{{\mathrm{fin}}}

\newcommand\Part[1]{{\mathcal P}({#1})}

\newcommand\Union{\bigcup}
\newcommand\Inter{\bigcap}

\newcommand{\Linarrow}{\multimap}

\newcommand\Myleft{}
\newcommand\Myright{}

\newcommand\Web[1]{\Myleft|{#1}\Myright|}

\newcommand\Mset[1]{[{#1}]}

\newcommand\Par[2]{{#1}\mathrel{\IPar}{#2}}

\newcommand\ITens{\otimes}

\newcommand\Tens[2]{{#1}\ITens{#2}}
\newcommand\Tensi[3]{{#2}\ITens_{#1}{#3}}
\newcommand\Tensp[2]{\left({#1}\ITens{#2}\right)}
\newcommand\IWith{\mathrel{\&}}
\newcommand\With[2]{{#1}\IWith{#2}}
\newcommand\IPlus{\oplus}
\newcommand\Plus[2]{{#1}\IPlus{#2}}
\newcommand\Orth[2][]{#2^{\Fbot_{#1}}}
\newcommand\Orthp[2][]{(#2)^{\Fbot_{#1}}}

\newcommand\Bwith{\mathop{\&}}
\newcommand\Bplus{\mathop\oplus}

\newcommand\Inj[1]{\overline\pi_{#1}}

\newcommand\Biorth[1]{#1^{\Fbot\Fbot}}

\newcommand\Triorth[1]{{#1}^{\Fbot\Fbot\Fbot}}

\newcommand\One{1}
\newcommand\Onelem{*}

\newcommand\Seq[1]{\vdash{#1}}

\newcommand\IExcl{{\mathord{!}}}

\newcommand\Locun[1]{1^J}

\newcommand\Isom\simeq

\newcommand\Comp{\mathrel\circ}

\newcommand\Funinv[1]{#1^{-1}}

\newcommand\Limpl[2]{{#1}\Linarrow{#2}}
\newcommand\Limplp[2]{\left({#1}\Linarrow{#2}\right)}

\newcommand\Nat{{\mathbb{N}}}

\newcommand\True{\mathsf{true}}
\newcommand\False{\mathsf{false}}

\newcommand\Zero{0}

\newcommand\List[3]{#1_{#2},\dots,#1_{#3}}

\newcommand\Subst[3]{{#1}\left[{#2}/{#3}\right]}

\newcommand\Substbis[2]{{#1}\left[{#2}\right]}

\newcommand\Mfin[1]{\mathcal M_\Fini({#1})}

\newcommand\Evlin{\operatorname{\mathsf{ev}}}
\newcommand\REL{\operatorname{\mathbf{Rel}}}

\newcommand\Redst[1]{\mathop{\mathsf{Red}}}

\newcommand\Tuple[1]{\langle{#1}\rangle}

\newcommand\Msetofsubst[1]{\bar F}

\newcommand\Matapp[2]{{#1}\Compl{#2}}
\newcommand\Matappa[2]{{#1}\cdot{#2}}
\newcommand\Matappap[2]{\left({#1}\cdot{#2}\right)}

\newcommand\Leftu{\lambda}
\newcommand\Rightu{\rho}
\newcommand\Assoc{\alpha}
\newcommand\Sym{\gamma}

\newcommand\Retri\zeta
\newcommand\Retrp\rho

\newcommand\Impl[2]{{#1}\Rightarrow{#2}}

\newcommand\Tsem[1]{\llbracket{#1}\rrbracket}
\newcommand\Tsemr[1]{\llbracket{#1}\rrbracket^{\REL}}
\newcommand\Tsemn[1]{\llbracket{#1}\rrbracket^{\NUTS}}
\newcommand\Psemr[1]{\llbracket{#1}\rrbracket^{\REL}}
\newcommand\Psemn[1]{\llbracket{#1}\rrbracket^{\NUTS}}

\newcommand\Psem[1]{\llbracket{#1}\rrbracket}

\newcommand\Tnat\iota

\newcommand\Loop\Omega

\newcommand\Timpl\Impl
\newcommand\Simpl\Impl

\newcommand\Emptyseq{()}

\newcommand\Weak[1]{\operatorname{\mathsf{w}}_{#1}}

\newcommand\Contr[1]{\operatorname{\mathsf{contr}}_{#1}}

\newcommand\Der[1]{\operatorname{\mathsf{der}}_{#1}}

\newcommand\Digg[1]{\operatorname{\mathsf{dig}}_{#1}}

\newcommand\Id{\operatorname{\mathsf{Id}}}

\newcommand\Proj[1]{\pi_{#1}}

\newcommand\Excl[1]{\oc{#1}}

\newcommand\Exclp[1]{\oc({#1})}

\newcommand\Int[1]{\wn{#1}}

\newcommand\Prom[1]{#1^!}

\newcommand\Relincl\eta
\newcommand\Relrestr\rho

\newcommand\Seelyz{\mathsf m^0}
\newcommand\Seelyt{\mathsf m^2}

\newcommand\Monoidalz{\mu^0}
\newcommand\Monoidalt{\mu^2}

\newcommand\Compl{\,}
\newcommand\Curlin{\mathsf{cur}}

\newcommand\Op[1]{{#1}^{\mathsf{op}}}

\newcommand\Em[1]{{#1}^\oc}

\newcommand\Snum[1]{\overline{#1}}

\newcommand\Vect[1]{\vec{#1}}

\newcommand\Bnfeq{\mathrel{\mathord:\mathord=}}
\newcommand\Bnfor{\,\,\mathord|\,\,}

\newcommand\Tot[1]{\mathop{\mathsf{Tot}}(#1)}

\newcommand\Vcsnot[1]{\mathbb{#1}}
\newcommand\Vcstnot[1]{\mathbb{#1}}
\newcommand\Vsnot[1]{\mathbb{#1}}

\newcommand\Strfun[1]{\overline{#1}}
\newcommand\Strnat[1]{\widehat{#1}}

\newcommand\ALGFUN[2]{\mathbf{Alg}_{#1}(#2)}
\newcommand\COALGFUN[2]{\mathbf{Coalg}_{#1}(#2)}

\newcommand\Lfpll[2]{\mu#1\,#2}
\newcommand\Gfpll[2]{\nu#1\,#2}
\newcommand\LLvars{\mathcal V}

\newcommand\Oneelem{*}

\newcommand\Ncut{(\textsf{cut})}

\newcommand\Ntens{($\ITens$)}

\newcommand\Npar{($\IPar$)}

\newcommand\Ncontr{(\textsf{c})}

\newcommand\Nlfp{($\mu$-\textsf{fold})}
\newcommand\Ngfp{($\nu$-\textsf{rec})}

\newcommand\Ngfpbis{($\nu$-\textsf{rec}${}'$)}
\newcommand\Ordinals{\mathbb O}

\newcommand\MULL{\mu\mathsf{LL}}
\newcommand\MUMALL{\mu\mathsf{MALL}}
\newcommand\MUALL{\mu\mathsf{ALL}}
\newcommand\LL{\mathsf{LL}}

\newcommand\Sone{1}
\newcommand\Sbot{\mathord\bot}

\tikzset{cong/.style={draw=none,edge node={node [sloped, allow upside down, auto=false]{$\cong$}}},
  Isom/.style={draw=none,every to/.append style={edge node={node [sloped, allow upside down, auto=false]{$\cong$}}}}}

\newcommand\isom{\mathrel{\stackon[-0.1ex]{\makebox*{\scalebox{1.08}{\AC}}{=\hfill\llap{=}}}{{\AC}}}}
\newcommand\nvisom{\rotatebox[origin=cc] {-90}{$ \isom $}}
\newcommand\visom{\rotatebox[origin=cc] {90} {$ \isom $}}

       \newcommand\Coalgca[1]{\underline{#1}}
\newcommand\Coalgstr[1]{h_{#1}}

\newcommand\LCAT{\cL}

\newcommand\Obj[1]{\mathsf{Obj}(#1)}
\newcommand\Promp[1]{\Prom{(#1)}}

\newcommand\Fcomod[1]{\mathsf{fc}_{#1}}

\newcommand\Klp[2]{{#1}_{#2}}

\newcommand\Mtens[3]{#2\otimes_{#1}#3}

\newcommand\Mevlin[1]{\Evlin_{#1}}

\newcommand\Kcomod[2]{#1[#2]}

\newcommand\Morth[2]{{#2}^{\Fbot[#1]}}
\newcommand\Mbiorth[2]{{#2}^{\Fbot[#1]\Fbot[#1]}}

\newcommand\Mexcl[2]{\oc_{#1}{#2}}
\newcommand\Mder[1]{\Der{}[#1]}
\newcommand\Mdigg[1]{\Digg{}[#1]}

\newcommand\Mfun[2]{{#1}[#2]}

\newcommand\Strid{\cI}
\newcommand\Strcst[1]{\cK^{#1}}

\newcommand\Fungfp[1]{\nu#1}
\newcommand\Funlfp[1]{\mu#1}
\newcommand\Funfp[1]{\sigma#1}

\newcommand\SF{\mathsf{F}}

\newcommand\Eset[1]{\left\{#1\right\}}

\newcommand\Excll[1]{\oc\oc{#1}}

\newcommand\RELI{\REL^{\subseteq}}
\newcommand\Relii{\eta^+}
\newcommand\Relip{\eta^-}
\newcommand\Botlin{\mathord\perp}

\newcommand\Upcl[1]{\mathord\uparrow#1}
\newcommand\Total[1]{\mathcal T(#1)}

\newcommand\NUTS{\mathbf{Nuts}}

\newcommand\Promhc[1]{#1^{(\mathord\oc)}}
\newcommand\Prommhc[1]{#1^{(\mathord\oc\mathord\oc)}}

\newcommand\VNUTS[1]{\mathbf{Vnuts}_{#1}}

\newcommand\VREL[1]{\REL_{#1}}

\newcommand\Nutsm{\mathsf p}
\newcommand\Nutsf{\mathsf u}

\newcommand\Leaf[1]{\langle#1\rangle}
\newcommand\Bnode[2]{\langle#1,#2\rangle}

\newcommand\Fseq[1]{{#1}^{<\omega}}
\newcommand\Klfree[1]{\mathsf F_{#1}}

\begin{document}
\title{Categorical models of Linear Logic with fixed points of
  formulas} 

\author{Thomas Ehrhard and Farzad Jafarrahmani\\
Université de Paris, CNRS, IRIF, F-75006, Paris, France}


\maketitle

\begin{abstract}
  We develop a categorical semantics of $\MULL$, a version of
  propositional Linear Logic with least and greatest fixed points
  extending David Baelde's propositional $\MUMALL$ with
  exponentials. Our general categorical setting is based on 
  Seely categories and on strong functors acting on them. We exhibit
  two simple instances of this setting. In the first one, which is
  based on the category of sets and relations, least and greatest
  fixed points are interpreted in the same way. In the second one,
  based on a category of sets equipped with a notion of totality
  (non-uniform totality spaces) and relations preserving it, least
  and greatest fixed points have distinct interpretations. This latter
  model shows that $\MULL$ enjoys a denotational form of normalization
  of proofs.
\end{abstract}

\section{Introduction}
Propositional Linear Logic is a well-established logical system
introduced by Girard in~\cite{Girard87}. It provides a fine-grain
analysis of proofs in intuitionistic and classical logic, and more
specifically of their cut-elimination. $\LL$ features a
logical account of the structural rules (weakening,
contraction) which are handled implicitly in intuitionistic and
classical logic. For this reason, $\LL$ has many useful outcomes in the
Curry-Howard based approach to the theory of programming: logical
understanding of evaluation strategies, new syntax
of proofs/programs (proof-nets), connections with other branches of
mathematics (linear algebra, functional analysis, differential calculus), new
operational semantics (geometry of interaction) etc.

However propositional $\LL$ is not a reasonable programming language,
by lack of data-types and iteration or recursion principles. This is
usually remedied by extending propositional $\LL$ to the
$2^{\mathrm{nd}}$ order, thus defining a logical system in which
Girard's System $\SF$~\cite{Girard89} can be embedded.
Another option to turn propositional $\LL$ into a programming language
--~closer to usual programming~-- is to extend it
with least and greatest fixed points of formulas. Such an extension
was early suggested by Girard in an unpublished note~\cite{Girard92},
though the first comprehensive proof-theoretic investigation of such
an extension of $\LL$ is recent: in~\cite{Baelde12} Baelde considers an
extension $\MUMALL$ of Multiplicative Additive $\LL$ sequent calculus
with least and greatest fixed points.
His motivations arose from a proof-search and system verification
perspective and therefore his $\MUMALL$ logical system is a predicate
calculus. Our purpose is to develop a more Curry-Howard oriented point
of view on $\LL$ with fixed points and therefore we stick to the
proposition calculus setting of~\cite{Girard89}. But,
unlike~\cite{Baelde12} we include the exponentials in our system from
the beginning%
\footnote{Exponentials are not considered in $\MUMALL$ because some
  form of exponential can be encoded using inductive/coinductive
  types, however these exponentials are not fully satisfactory from
  our point of view because their denotational interpretation does not
  satisfy all required isomorphisms; specifically, the \emph{Seely
    isos} are lacking.}%
, so we call it $\MULL$ rather than propositional $\MUMALL$ and we
consider it as an alternative to the ``system $\SF$'' approach to
representing programs in $\LL$. Our system $\MULL$ could also have
applications to session types, in the line of~\cite{LindleyMorris16}.
The $\nu$-introduction rule of $\MULL$ (Park's rule, that is
rule~\Ngfp{} of Section~\ref{sec:MULL-syntax}) leads to subtle
cut-elimination rewrite rules for which Baelde could prove
cut-elimination in $\MUMALL$, showing for instance that a proof of the
type of integers $\Lfpll\zeta{(\Plus\One\zeta)}$ necessarily reduces
to an integer (in contrast with $\LL$, $\MUMALL$ enjoys
only a restricted form of sub-formula property). There are alternative
proof-systems for the same logic, involving infinite or cyclic proofs,
see~\cite{BaeldeDoumaneSaurin16}, whose connections with the
aforementioned finitary proof-system are not completely clear yet.

Since the proof-theory (and hence the ``operational semantics'') of
$\MULL$ is still under development, it is important to investigate its
categorical semantics, whose definition does not rely on the precise
choice of inference and rewrite rules we equip $\MULL$ with, see the
\emph{Outcome} § below. We develop here a categorical semantics of
$\MULL$ extending the standard notion of Seely category%
\footnote{Sometimes called new-Seely category: it is a cartesian
  symmetric monoidal closed category with a $\ast$-autonomous
  structure and a comonad $\Excl\_$ with a strong symmetric monoidal
  structure from the cartesian product to the tensor product.}%
of classical $\LL$, see~\cite{Mellies09}. Such a model of $\MULL$
consists of a Seely category $\cL$ and of a class of functors
$\cL^n\to\cL$ for all possible arities $n$ which will be used for
interpreting $\MULL$ formulas with free variables. These functors have
to be equipped with a strength to deal properly with contexts in the
rule~\Ngfp, see Section~\ref{sec:MULL-comments} for a discussion on
these contexts in particular.

Then we develop a simple instance of this setting which consists in
taking for $\cL$ the category of sets and relations, a well-known
Seely model of $\LL$. The \emph{variable sets} are the strong functors
we consider on this category. They are the pairs
$\Vsnot F=(\Strfun{\Vsnot F},\Strnat{\Vsnot F})$ where
$\Strnat{\Vsnot F}$ is the strength and
$\Strfun{\Vsnot F}:\REL^n\to\REL$ is a functor which is
Scott-continuous in the sense that it commutes with directed unions of
morphisms. This property implies that $\Strfun{\Vsnot F}$ maps
injections to injections and is cocontinuous on the category of sets
and injections.
There is no special requirement about the strength $\Strnat{\Vsnot F}$
beyond naturality, monoidality and compatibility with the
comultiplication of the comonad $\Excl\_$. Variable sets form a Seely
model of $\MULL$ where linear negation is the identity on objects. The
formulas $\Lfpll\zeta F$ and $\Gfpll\zeta F$ are interpreted as the
same variable set, exactly as $\ITens$ and $\IPar$ are interpreted in
the same way (and similarly for additives and exponentials). This
denotational ``degeneracy'' at the level of types is a well known
feature of $\REL$ which does not mean at all that the model is trivial.
For instance normal multiplicative exponential $\LL$ proofs which have
distinct relational interpretations have distinct associated
proof-nets~\cite{CarvalhoDeFalco12,Carvalho16}.

Last we enrich this model $\REL$ by considering sets equipped with an
additional structure of \emph{totality}: a \emph{non-uniform totality
  space} (NUTS) is a pair $X=(\Web X,\Total X)$ where $\Web X$ is a
set and $\Total X$ is a set of subsets of $\Web X $which intuitively
represent the total, that is, terminating computations of type
$X$. This set $\Total X$ is required to coincide with its bidual for a
duality expressed in terms of non-empty intersections.
This kind of definition by duality is ubiquitous in $\LL$
since~\cite{Girard87} and has been categorically formalized as
\emph{double gluing} in~\cite{HylandSchalk03}. We don't use this
categorical formalization here however as it would not simplify the
presentation.
One nice feature of this specific duality is that the bidual of a set
of subsets of $\Web X$ is simply its upwards-closure
(wrt.~inclusion)\footnote{This new model is a major simplification
  wrt.~notions of totality on coherence spaces~\cite{Girard86} or
  Loader's totality spaces~\cite{Loader94} where biduality is much
  harder to deal with because it combines totality with a form of
  determinism.}, see
Lemma~\ref{lemma:NUTS-biorth-uppper-closed}. Given two NUTS $X$ and
$Y$ there is a natural notion of \emph{total relation}
$t\subseteq\Web X\times\Web Y$ giving rise to a category $\NUTS$ which
is easily seen to be a Seely model of $\LL$. To turn it into a
categorical model of $\MULL$, we need a notion of strong functors
$\NUTS^n\to\NUTS$. Rather than considering them directly as functors,
we define \emph{variable non-uniform totality spaces} (VNUTS) as pairs
$\Vsnot X=(\Web{\Vsnot X},\Total{\Vsnot X})$ where
$\Web{\Vsnot X}:\REL^n\to\REL$ is a variable set and, for each tuple
$\Vect X=(\List X1n)$ of VNUTS's, $\Total{\Vsnot X}(\Vect X)$ is a
totality structure on the set
$\Strfun{\Web{\Vsnot X}}(\Web{\Vect X})$. It is also required that the
action of the functor $\Strfun{\Web{\Vsnot X}}$ on $\NUTS$ morphisms
and the strength $\Strnat{\Vsnot X}$ respect this totality
structures. Then it is easy to derive from such a VNUTS $\Vsnot X$ a
strong functor $\NUTS^n\to\NUTS$ and we prove that, equipped with
these strong functors, $\NUTS$ is a model of $\MULL$.

\paragraph{Outcome} %
One major benefit of this construction is that it gives a value to all
proofs of $\MULL$, invariant by cut-elimination. Moreover, the fact
that this value is total shows in a \emph{syntax independent way} that
when $\pi$ is for instance a $\MULL$ proof of $\Plus\One\One$ (the
type of booleans), the value associated with $\pi$ is non-empty, that
is, $\pi$ has a defined boolean value $\True$ or $\False$%
\footnote{Or both because our $\NUTS$ model accepts
  non-determinism. By adding a \emph{non-uniform} coherence relation
  as defined in~\cite{BucciarelliEhrhard99,Boudes11} to the model one
  can show that this value is actually a uniquely defined boolean. See
  also Section~\ref{sec:example-integers}.}%
.  We could also obtain this by a normalization theorem: $\pi$ reduces
to one of the two normal proofs of $\Plus\One\One$ (and if we prove
for instance a Church-Rosser theorem we will know that this proof is
unique). Such proofs would depend of course on the actual presentation
of the syntax whereas our denotational argument does not.

\paragraph{Related work}%
There is a vast literature on extending logic with fixed point that we
cannot reasonably summarize, see the discussions
in~\cite{DoumanePHD,BaeldeDoumaneSaurin16}.  Cut-elimination of such
systems has been extensively investigated, see for
instance~\cite{BrotherstonSimpson11,MomiglianoTiu12,McDowellMiller00,CamposFiore20}.
Closer to ours is the work of Santocanale~\cite{Santocanale02} and its
categorical interpretation in $\mu$-bicomplete
categories~\cite{FortierSantocanale13} which, unlike most
contributions in this field, considers also categorical
interpretations of proofs. Santocanale \emph{et al.}~consider
circular proofs whereas we use Park's rule. A deeper difference
lies in the logic itself: from an $\LL$ point of view the logic
considered by Santocanale \emph{et al.}~is \emph{purely additive
  linear logic with least and greatest fixed points $\MUALL$} which
seems too weak in our Curry-Howard perspective. And indeed
$\mu$-bicomplete categories do not provide the monoidal and
exponential structures required for interpreting $\MULL$.

In~\cite{Loader97}, that we became aware of only recently (and seems
related to the earlier report~\cite{Geuvers92}), Loader extends the
simply typed $\lambda$-calculus with inductive types and develops its
denotational semantics. His models are cartesian closed categories
$\cC$ equipped with a class of strong functors and seem very close to
ours (Section~\ref{sec:cat-models}): one might think that any of our
models yields a Loader model as its Kleisli category. This is not the
case because in a Loader model the category $\cC$ is
cocartesian%
\footnote{To account for the disjunction of his logical
  system which is crucial for defining interesting data-types such as the
  integers.} %
whereas the Kleisli category of a Seely category is not
cocartesian in general: this would require to have an iso between
$\Exclp{\Plus XY}$ and $\Plus{\Excl X}{\Excl Y}$ which is usually
absent.
%
Loader studies two concrete instances of his models: one is based on
recursion theory (partial equivalence relations) and the other on a
notion of domains with totality described as a model of $\LL$. This model
might give rise to one of our Seely models, this point requires
further studies.
Our NUTS are quite different from Loader totality
domains which feature a notion of ``consistency'' enforcing some kind
of determinism and, combined with totality, allow the Kleisli category
to be cocartesian as well.  Our model is based on $\REL$ and therefore
is compatible with
non-determinism~\cite{BucciarelliEhrhardManzonetto12} and PCF
recursion. This is important for us because we would like to consider
rules beyond Park's rule for inductive and coinductive types, based on
PCF fixed points --~with further guardedness conditions for
  guaranteeing termination~-- in the spirit
of~\cite{Coquand93,Paulin-Mohring93,Gimenez98} or even on infinite
terms in the spirit of~\cite{BaeldeDoumaneSaurin16}.

We mention also the work of
Clairambault~\cite{Clairambault09,Clairambault13} who investigates the
game with totality semantics of an extension of intuitionistic logic
with least and greatest fixed points (independently
of~\cite{Loader97,Geuvers92}). A Kleisli-like connection with his work
should be sought too.

\paragraph{Notations}
We use the following conventions: $\Vect a$ stands for a
list $(\List a1n)$. A unary operation $f$ is extended to lists of
arguments in the obvious way: $f(\Vect a)=(f(a_1),\dots,f(a_n))$. When
we write natural transformations, we very often omit the objects where
they are taken keeping them implicit for the sake of readability,
when they can easily be retrieved from the context. If $\cA$ is a
category then $\Obj\cA$ is its class of objects and if $A,B\in\Obj\cA$
then $\cA(A,B)$ is the set of morphisms from $A$ to $B$ in $\cA$ (all
the categories we consider are locally small). If
$\cF:\cA\times\cB\to\cC$ is a functor and $A\in\Obj\cA$ then
$\cF_A:\cB\to\cC$ is the functor defined by $\cF_A(B)=\cF(A,B)$ and
$\cF_A(f)=\cF(A,f)$ (we often write $A$ instead of $\Id_A$).

Most proofs can be found in an Appendix.

\section{Categorical models of $\LL$}\label{sec:LL-models}

\subsection{Seely categories.}
\label{sec}
We recall the basic notion of categorical model of $\LL$. Our main
reference is the notion of a \emph{Seely category} as presented
in~\cite{Mellies09}. We refer to that survey for all the technical
material that we do not recall here.

A Seely category is a symmetric monoidal closed category (SMCC)
$(\LCAT,\ITens,\One,\Leftu,\Rightu,\Assoc,\Sym)$ where
$\Leftu_X\in\cL(\Tens\One X,X)$, $\Rightu_X\in\cL(\Tens X\One,X)$,
$\Assoc_{X,Y,Z}\in\cL(\Tens{(\Tens XY)}{Z},\Tens{X}{(\Tens{Y}{Z})})$
and $\Sym_{X,Y}\in\cL(\Tens XY,\Tens YX)$ are natural isomorphisms
satisfying coherence diagrams that we do not record here. We use
$\Limpl XY$ for the object of linear morphisms from $X$ to $Y$,
$\Evlin\in\LCAT(\Tens{(\Limpl XY)}{X},Y)$ for the evaluation morphism
and $\Curlin$ for the linear curryfication map
$\LCAT(\Tens ZX,Y)\to\LCAT(Z,\Limpl XY)$. We assume $\LCAT$ to be
$\ast$-autonomous with dualizing object $\Sbot$ (this object is part of
the structure of a Seely category). We use $\Orth X$ for the object
$\Limpl X\Sbot$ of $\LCAT$ (the dual, or linear negation, of $X$).
It is also assumed that $\LCAT$ is cartesian with final object $\Top$,
product $\With {X_1}{X_2}$ with projections $\Proj1,\Proj2$. By
$\ast$-autonomy $\LCAT$ is cocartesian with initial object $\Zero$,
coproduct $\IPlus$ and injections $\Inj i$.

We also assume to be given a comonad $\Excl\_:\LCAT\to\LCAT$ with
counit $\Der X\in\LCAT(\Excl X,X)$ (\emph{dereliction}) and
comultiplication $\Digg X\in\LCAT(\Excl X,\Excl{\Excl X})$
(\emph{digging}) together with a strong symmetric monoidal structure
(Seely natural isos $\Seelyz:\Sone\to\Excl\Top$ and $\Seelyt$ with
$\Seelyt_{X_1,X_2}:\Tens{\Excl{X_1}}{\Excl{X_2}}\to\Exclp{\With{X_1}{X_2}}$
for the functor $\Excl\_$, from the symmetric monoidal
category $(\LCAT,\IWith)$ to the symmetric monoidal category
$(\LCAT,\ITens)$ satisfying an additional coherence condition
wrt.~$\Digg{}$).
This strong monoidal structure allows to define a lax monoidal
structure $(\Monoidalz,\Monoidalt)$ of $\Excl\_$ from $(\LCAT,\ITens)$
to itself. More precisely $\Monoidalz\in\LCAT(\Sone,\Excl\Sone)$ and
$\Monoidalt_{X_1,X_2}\in
\LCAT(\Tens{\Excl{X_1}}{\Excl{X_2}},\Exclp{\Tens{X_1}{X_2}})$ are
defined using $\Seelyz$, $\Seelyt$, $\Der{}$ and $\Digg{}$ (and are
not isos in most cases). Also, for each object $X\in\Obj\cL$, there is
a canonical structure of commutative $\ITens$-comonoid on $\Excl X$
given by $\Weak X\in\cL(\Excl X,\One)$ and
$\Contr X\in\cL(\Excl X,\Tens{\Excl X}{\Excl X})$. The definition of
these morphisms involves all the structure of $\Excl\_$ explained
above, and in particular the Seely isos.
%
We use $\Int\_$ for the ``De Morgan dual'' of $\Excl\_$:
$\Int X=\Orthp{\Exclp{\Orth X}}$ and similarly for morphisms.

\subsection{Oplax monoidal comonads}
\label{sec:oplax-monoidal}
Let $\cM$ be a symmetric monoidal category (with the same notations as
above for the tensor product) and $(T,\epsilon,\mu):\cM\to\cM$ be a
comonad ($\epsilon$ is the unit and $\mu$ the multiplication). An
\emph{oplax monoidal} structure on $T$ consists of a morphism
$\theta^0\in\cM(T\Sone,\Sone)$ and a natural transformation
$\theta^2_{X_1,X_2}\in\cM(T(\Tens{X_1}{X_2}),\Tens{T(X_1)}{T(X_2)})$
subject to standard symmetric monoidality and compatibility with
$\epsilon$ and $\mu$, this latter reading
$\Tensp{\epsilon_{X_1}}{\epsilon_{X_2}}\Compl\theta_{X_1,X_2}
=\epsilon_{\Tens{X_1}{X_2}}$ and:
\begin{center}
  \begin{tikzcd}
    T(\Tens{X_1}{X_2})
    \arrow[r,"\theta_{X_1,X_2}"]
    \arrow[d,"\mu_{\Tens{X_1}{X_2}}"]
    &[-0.3em]\Tens{TX_1}{TX_2}
    \arrow[r,"\Tens{\mu_{X_1}}{\mu_{X_2}}"]
    &[0.8em]\Tens{T^2X_1}{T^2X_2}\\
    T^2(\Tens{X_1}{X_2})
    \arrow[rr,"T(\theta_{X_1,X_2})"]
    &&T(\Tens{TX_1}{TX_2})
    \arrow[u,"\theta_{TX_1,TX_2}"]
  \end{tikzcd}
\end{center}
Then the Kleisli category $\cM_T$ has a canonical symmetric monoidal
structure, with unit $\Sone$ and tensor product $\Tens{X_1}{X_2}$
defined as in $\cM$ for objects. Given
$f_i\in\cM_T(X_i,Y_i)$,
$\Tensi T{f_1}{f_2}\in\cM_T(\Tens{X_1}{X_2},\Tens{Y_1}{Y_2})$ is defined as
\begin{center}
  \begin{tikzcd}
    T(\Tens{X_1}{X_2})
    \arrow[r,"\theta^2_{X_1,X_2}"]
    &\Tens{TX_1}{TX_2}
    \arrow[r,"\Tens{f_1}{f_2}"]
    &\Tens{Y_1}{Y_2}
  \end{tikzcd}.
\end{center}
Let $\Klfree T:\cM\to\cM_T$ be the canonical functor which acts as the
identity on objects and maps $f\in\cM(X,Y)$ to
$f\Compl\epsilon_X\in\cM_T(X,Y)$.

\subsection{Eilenberg-Moore category and free comodules}
\label{sec:EM-Kl-category}
Let $\LCAT$ be a Seely category.  Since $\Excl\_$ is a comonad we can
define the category $\Em\LCAT$ of $\IExcl$-coalgebras (Eilenberg-Moore
category of $\Excl\_$). An object of this category is a pair
$P=(\Coalgca P,\Coalgstr P)$ where $\Coalgca P\in\Obj\LCAT$ and
$\Coalgstr P\in\LCAT(\Coalgca P,\Excl{\Coalgca P})$ is such that
$\Der{\Coalgca P}\Compl\Coalgstr P=\Id$ and
$\Digg{\Coalgca P}\Compl\Coalgstr P=\Excl{\Coalgstr P}\Compl\Coalgstr
P$.  Then $f\in\Em\LCAT(P,Q)$ if $f\in\LCAT(\Coalgca P,\Coalgca Q)$
and $\Coalgstr Q\Compl f=\Excl f\Compl\Coalgstr P$.  The functor
$\Excl\_$ can be seen as a functor from $\LCAT$ to $\Em\LCAT$ mapping
$X$ to $(\Excl X,\Digg X)$ and $f\in\LCAT(X,Y)$ to $\Excl f$. It is
right adjoint to the forgetful functor $\Em\LCAT\to\LCAT$. Given
$f\in\LCAT(\Coalgca P,X)$, we use $\Prom f\in\Em\LCAT(P,\Excl X)$ for
the morphism associated with $f$ by this adjunction, one has
$\Prom f=\Excl f\Compl\Coalgstr P$. If $g\in\Em\LCAT(Q,P)$, we have
$\Prom f\Compl g=\Promp{f\Compl g}$.
%
Then $\Em\LCAT$ is cartesian with final object
$(\One,\Coalgstr\One=\Monoidalz)$ still denoted as $\One$ and product
$\Tens{P_1}{P_2}=(\Tens{\Coalgca{P_1}}{\Coalgca{P_2}},
\Coalgstr{\Tens{P_1}{P_2}})$ with
$\Coalgstr{\Tens{P_1}{P_2}}:
\begin{tikzcd}
  \Tens{\Coalgca{P_1}}{\Coalgca{P_2}}
  \arrow[r,"\Tens{\Coalgstr{P_1}}{\Coalgstr{P_2}}"]
  &[1em]\Tens{\Excl{\Coalgca{P_1}}}{\Excl{\Coalgca{P_2}}}
  \arrow[r,"\Monoidalt_{\Coalgca{P_1},\Coalgca{P_2}}"]
  &[0.4em]\Exclp{\Tens{\Coalgca{P_1}}{\Coalgca{P_2}}}
\end{tikzcd}$.  
This category is also cocartesian with initial object
$(\Zero,\Coalgstr\Zero)$ still denoted as $\Zero$ and coproduct
$\Plus{P_1}{P_2}
=(\Plus{\Coalgca{P_1}}{\Coalgca{P_2}},\Coalgstr{\Plus{P_1}{P_2}})$
with $\Coalgstr{\Plus{P_1}Q}$ defined as follows. For $i=1,2$ one
defines
$h^i:\Coalgca{P_i}\to\Exclp{\Plus{\Coalgca{P_1}}{\Coalgca{P_2}}}$ as
\begin{tikzcd}
  \Coalgca P_1
  \arrow[r,"\Coalgstr{P_1}"]
  &[-1em]\Excl{\Coalgca{P_1}}
  \arrow[r,"\Excl{\Inj i}"]
  &[-1em]\Exclp{\Plus{\Coalgca{P_1}}{\Coalgca{P_2}}}
\end{tikzcd}
and then $\Coalgstr{\Plus{P_1}{P_2}}$ is the unique
morphism
$\Plus{\Coalgca{P_1}}{\Coalgca{P_2}}\to\Exclp{\Plus{\Coalgca{P_1}}{\Coalgca
    {P_2}}}$ such that $\Coalgstr{\Plus{P_1}{P_2}}\Compl\Inj i=h_i$
for $i=1,2$.
More details can be found in~\cite{Mellies09}. We use
$\Contr P\in\Em\LCAT(P,\Tens PP)$ (\emph{contraction}) for the
diagonal and $\Weak P\in\Em\LCAT(P,\One)$ (\emph{weakening}) for the
unique morphism to the final object.

\subsubsection{The $\LL$ model of free comodules on a given
  coalgebra}\label{sec:free-comodules-model}
Given an object%
\footnote{In this paper we could restrict to the case where $P$ is a
  tensor of ``free coalgebras'' $(\Excl{X_i},\Digg{X_i})$ but it is
  more natural to deal with the general case, which will be quite
  useful in further work, see Section~\ref{sec:conclusion}.}
$P$ of $\Em\LCAT$, we can define a functor
$\Fcomod P:\LCAT\to\LCAT$ which maps an object $X$ to
$\Tens{\Coalgca P}{X}$ and a morphism $f$ to $\Tens{\Coalgca
  P}{f}$. This functor is clearly an oplax monoidal comonad (with
structure maps defined using $\Weak P$, $\Contr P$ and the monoidal
structure of $\LCAT$)\footnote{The definition of this comonad uses
  only the comonoid structure of $\Coalgca P$. The $\Excl\_$-structure
  will be used later.}.  A coalgebra for this comonad is a
\emph{$P$-comodule}.
By Section~\ref{sec:oplax-monoidal} the Kleisli category
$\Kcomod\LCAT P=\Klp\LCAT{\Fcomod P}$ of this comonad (that is, the
category of free $P$-comodules) has a canonical structure of symmetric
monoidal category (SMC).  We set
$\Klfree P=\Klfree{\Fcomod P}:\LCAT\to\Kcomod\LCAT P$.  Girard showed
in~\cite{Girard98c}
%
%
that $\Kcomod\LCAT P$ is a Seely model of $\LL$ with operations on
objects defined in the same way as in $\LCAT$, and using the coalgebra
structure of $P$ for the operations on morphisms. Intuitively, $P$
should be considered as a given context and $\Kcomod\LCAT P$ as a
model in this context. This idea appears at various places in the
literature, see for instance~\cite{CurienFioreMunch16,UustaluVene08}.
Let us summarize this construction.
If $f_i\in\Kcomod\LCAT P(X_i,Y_i)$ for $i=1,2$ then
$\Mtens P{f_1}{f_2}=\Tensi{\Fcomod P}{f_1}{f_2}\in\Kcomod\LCAT
P(\Tens{X_1}{X_2},\Tens{Y_1}{Y_2})$ is given by
\begin{center}
\begin{tikzcd}
  \Coalgca P\ITens X_1\ITens X_2
  \arrow[r,"\Contr{\Coalgca P}\ITens\Id"]
  &[1.4em]\Coalgca P\ITens \Coalgca P\ITens X_1\ITens X_2
  \arrow[d,phantom,"\nvisom"]\\[-1em]
  \Tens{Y_1}{Y_2}
  &[-1.2em]\Coalgca P\ITens X_1\ITens \Coalgca P\ITens X_2
  \arrow[l,swap,"\Tens{f_1}{f_2}"]
\end{tikzcd}  
\end{center}
The object of linear morphisms from $X$ to $Y$ in $\Kcomod\LCAT P$ is
$\Limpl XY$, and the evaluation morphism
$\Mevlin P\in\Kcomod\LCAT P(\Tens{(\Limpl XY)}{X},Y)$ is simply
$\Klfree P(\Evlin)$.
Then it is easy to check that if
$f\in\Kcomod\LCAT P(\Tens ZX,Y)$, that is
$f\in\LCAT(\Coalgca P\ITens Z\ITens X,Y)$, the morphism
$\Curlin(f)\in\Kcomod\LCAT P(Z,\Limpl XY)$ satisfies the required
monoidal closedness equations. With these definitions, the category
$\Kcomod\LCAT P$ is $\ast$-autonomous, with $\Sbot$ as dualizing
object. Specifically, given $f\in\Kcomod\LCAT P(X,Y)$, then $\Morth Pf$
is the following composition of morphisms:
\begin{center}
  \begin{tikzcd}
    \Tens{\Coalgca P}{\Orth Y}
    \arrow[r,"\Tens{\Coalgca P}{\Orth f}"]
    &\Coalgca P\ITens(\Limpl{\Coalgca P}{\Orth X})
    \arrow[r,"\Evlin"]
    &[-1em]\Orth X
  \end{tikzcd}
\end{center}
using implicitly the iso between $\Orthp{\Tens{Z}{X}}$ and
$\Limpl Z{\Orth X}$, and the $\ast$-autonomy of $\LCAT$ allows to prove
that indeed $\Mbiorth Pf=f$.
The category $\Kcomod\LCAT P$ is easily seen to be cartesian with
$\Top$ as final object, $\With{X_1}{X_2}$ as cartesian product (and
projections defined in the obvious way, applying $\Klfree P$ to the
projections of $\LCAT$). Last we define a functor
$\Mexcl P\_:\Kcomod\LCAT P\to\Kcomod\LCAT P$ by $\Mexcl PX=\Excl X$
and, given $f\in\Kcomod\LCAT P(X,Y)$, we define
$\Mexcl Pf\in\Kcomod\LCAT P(\Excl X,\Excl Y)$ as
\begin{tikzcd}
  \Tens{\Coalgca P}{\Excl X} \arrow[r,"\Tens{\Coalgstr P}{\Excl X}"]
  &\Tens{\Excl{\Coalgca P}}{\Excl X} \arrow[r,"\Monoidalt"]
  &[-1em]\Exclp{\Tens{{\Coalgca P}}{X}} \arrow[r,"\Excl f"]
  &[-1em]\Excl Y
\end{tikzcd}
and this functor has a comonad structure $(\Mder P,\Mdigg P)$ defined
by $\Mder P=\Klfree P(\Der{})$ and
$\Mdigg P=\Klfree P(\Digg{})$\footnote{The definition of $\Mexcl Pf$
  requires $P$ to be a $\oc$-coalgebra and not simply a commutative
  $\ITens$-comonoid. Of course if $\oc$ is the free exponential as
  in~\cite{Girard98c} the latter condition implies the former.}.

\begin{remark}
  Any $p\in\Em\LCAT(P,Q)$ induces a functor
  $\Kcomod\LCAT p:\Kcomod\LCAT Q\to\Kcomod\LCAT P$ which acts as the
  identity on objects and maps $f\in\Kcomod\LCAT Q(X,Y)$ to
  $\Kcomod\LCAT p(f)=f\Compl(\Tens pX)\in\Kcomod\LCAT P(X,Y)$. This
  functor is strict monoidal symmetric and preserves all the
  constructions of $\LL$, for instance
  $\Kcomod\LCAT p(\Mdigg Q)=\Mdigg P$ (simply because
  $\Kcomod\LCAT p\Comp\Klfree Q=\Klfree P$) and also
  $\Kcomod\LCAT p(\Mexcl Qf)=\Mexcl P{(\Kcomod\LCAT p(f))}$. We can
  actually consider $\Kcomod\LCAT\_$ as a functor from $\Op{\Em\LCAT}$
  to the category of Seely categories and functors which preserve
  their structure on the nose. This functor could probably more
  suitably be considered as a fibration in the line
  of~\cite{PowerRobinson97}, Section~7.
\end{remark}

\subsection{Strong functors on $\LCAT$}\label{sec:gen-strong-functors}

Given $n\in\Nat$, an \emph{$n$-ary strong functor} on $\LCAT$ is a
pair $\Vcsnot F=(\Strfun{\Vcsnot F},\Strnat{\Vcsnot F})$ where
$\Strfun{\Vcsnot F}:\LCAT^n\to\LCAT$ is a functor and
$\Strnat{\Vcsnot F}_{X,\Vect Y}\in\LCAT(\Tens{\Excl X}{\Strfun{\Vcsnot
    F}(\Vect Y)},\Strfun{\Vcsnot F}(\Tens{\Excl X}{\Vect Y}))$ is a
natural transformation, called the \emph{strength} of $\Vcsnot F$. We
use the notation
$\Tens Z{(\List Y1n)}=(\Tens{Z}{Y_1},\dots,\Tens{Z}{Y_n})$. It is
assumed moreover that the diagrams of
Figure~\ref{fig:strength-monoidality} commute, expressing the
monoidality of this strength as well as its compatibility with the
comultiplication of $\Excl\_$.
\begin{figure*}[t]
  {\footnotesize
    \begin{tikzcd}
      (\Excl{X_1}\ITens\Excl{X_2}) \ITens{\Strfun{\Vcsnot
          F}}(\Vect{Y}) \arrow[r,"\Tens{\Seelyt}{\Strfun{\Vcsnot
          F}(\Vect Y)}"] \arrow[d,swap,"\Tens{\Excl{X_1}}
      {\Strnat{\Vcsnot F}_{X_2,\Vect Y}}"]
      &[2.2em]\Exclp{\With{X_1}{X_2}}\ITens \Strfun{\Vcsnot F}(\Vect
      Y) \arrow[dd,"\Strnat{\Vcsnot F}_{\With{X_1}{X_2},\Vect Y}"]
      \\
      \Excl{X_1}\ITens \Strfun{\Vcsnot F}(\Excl{X_2}\ITens\Vect Y)
      \arrow[d,swap,"\Strnat{\Vcsnot F}_{X_1, \Tens{\Excl{X_2}}{\Vect
          Y}}"] &
      \\
      \Strfun{\Vcsnot F}(\Excl{X_1}\ITens\Excl{X_2} \ITens
      \Strfun{\Vcsnot F}(\Vect Y)) \arrow[r,"\Strfun{\Vcsnot
        F}(\Tens{\Seelyt}{\Vect Y})"] &\Strfun{\Vcsnot
        F}(\Exclp{\With{X_1}{X_2}}\ITens\Vect Y)
    \end{tikzcd}
    \quad
    \begin{tikzcd}
      \Tens{\Sone}{\Strfun{\Vcsnot F}(\Vect Y)}
      \arrow[r,"\Tens{\Seelyz}{\Strfun{\Vcsnot F}(\Vect Y)}"]
      \arrow[d,phantom,"\visom"]
      &[2.2em]\Tens{\Excl\Top}{\Strfun{\Vcsnot F}(\Vect Y)}
      \arrow[d,"\Strnat{\Vcsnot F}_{\Top,\Vect Y}"]
      \\
      \Strfun{\Vcsnot F}(\Tens\Sone{\Vect Y})
      \arrow[r,"\Strfun{\Vcsnot F}(\Tens{\Seelyz}{\Vect Y})"]
      &\Strfun{\Vcsnot F}(\Excl{\Top}\ITens\Vect Y)
    \end{tikzcd}
    \quad
    \begin{tikzcd}
      \Tens{\Excl X}{\Strfun{\Vcsnot F}(\Vect Y)}
      \arrow[r,"\Tens{\Digg X}{\Strfun{\Vcsnot F}(\Vect Y)}"]
      \arrow[d,swap,"\Strnat{\Vcsnot F}_{X,\Vect Y}"]
      &[2.2em]\Tens{\Excll X}{\Strfun{\Vcsnot F}(\Vect Y)}
      \arrow[d,"\Strnat{\Vcsnot F}_{\Excl X,\Vect Y}"]\\
      \Strfun{\Vcsnot F}(\Tens{\Excl X}{\Vect Y})
      \arrow[r,"\Vcsnot{\Vcsnot F}(\Tens{\Digg X}{\Vect Y})"]
      &\Strfun{\Vcsnot F}(\Tens{\Excll X}{\Vect Y})
    \end{tikzcd}
}
\caption{Monoidality and $\Digg{}$ diagrams for strong functors}
\label{fig:strength-monoidality}
\end{figure*}
The main purpose of this definition is that for any object $P$ of
$\Em\cL$ one can lift $\Vcsnot F$ to a functor
$\Mfun{\Vcsnot F}P:\Kcomod\LCAT P^n\to\Kcomod\LCAT P$ as
follows. First one sets
$\Mfun{\Vcsnot F}P(\Vect X)=\Strfun{\Vcsnot F}(\Vect X)$. Then, given
$\Vect f\in\Kcomod\LCAT P^n(\Vect X,\Vect Y)$ we define
$\Mfun{\Vcsnot F}{P}(\Vect f)\in\Kcomod\LCAT P({\Vcsnot F}(\Vect
X),{\Vcsnot F}(\Vect Y))$ as %

{\footnotesize
\begin{center}
  \begin{tikzcd}
    \Tens{\Coalgca P}{\Vcsnot F(\Vect X)}
    \arrow[r,"\Tens{\Coalgstr P}{\Id}"]
    &[1.2em]\Tens{\Excl{\Coalgca P}}{\Vcsnot F(\Vect X)}
    \arrow[r,"\Strnat{\Vcsnot F}"]
    &[-1em]\Vcsnot F(\Tens{\Excl{\Coalgca P}}{\Vect X})
    \arrow[r,"\Strfun{\Vcsnot F}(\Der{\Coalgca P}\ITens\Vect X)"]
    &[2em]\Vcsnot F(\Tens{{\Coalgca P}}{\Vect X})
    \arrow[d,swap,"\Vcsnot F(\Vect f)"]\\
    &&&\Vcsnot F(\Vect Y)
  \end{tikzcd}
\end{center}%
}%
\noindent
The fact that we have defined a functor results from the three
diagrams of Figure~\ref{fig:strength-monoidality} and from the
definition of $\Weak P$ and $\Contr P$ based on the Seely
isomorphisms.

\begin{remark}
  Since the seminal work of Moggi~\cite{Moggi89} strong functors play
  a central role in semantics for representing
  \emph{effects}.
  Our adaptation of this notion to the present $\LL$ setting follows
  the definition of an $\cL$-tensorial strength
  in~\cite{KobayashiS97}.
\end{remark}

\subsubsection{Operations on strong functors}%
\label{sec:strong-fun-LL-operations}%
There is an obvious unary identity strong functor $\Strid$ and for
each object $Y$ of $\LCAT$ there is an $n$-ary $Y$-valued constant
strong functor $\Strcst Y$; in the first case the strength natural
transformation is the identity morphism and in the second case, it is
defined using $\Weak{\Excl X}$. Let $\Vcsnot F$ be an $n$-ary strong
functor and $\List{\Vcsnot G}1n$ be $k$-ary strong functors. Then one
defines a $k$-ary strong functor
$\Vcsnot H=\Vcsnot F\Comp(\List{\Vcsnot G}1n)$: the functorial
component $\Strfun{\Vcsnot H}$ is defined in the obvious compositional
way. The strength is%

{\footnotesize
\begin{center}
  \begin{tikzcd}
    \Tens{\Excl X}{\Strfun{\Vcsnot H}(\Vect Y)}
    \arrow[r,"\Strnat{\Vcsnot F}"]
    &[-1em]\Strfun{\Vcsnot F}((\Tens{\Excl X}
    {\Strfun{\Vcsnot G_i}(\Vect Y)})_{i=1}^n)
    \arrow[r,"\Strfun{\Vcsnot F}((\Strnat{\Vcsnot G_i})_{i=1}^k)"]
    &[1.2em]\Strfun{\Vcsnot F}((
    {\Strfun{\Vcsnot G_i}(\Tens{\Excl X}{\Vect Y})})_{i=1}^n)
  \end{tikzcd}
\end{center}
}%
\noindent%
and is easily seen to satisfy the commutations of
Figure~\ref{fig:strength-monoidality}.  Given an $n$-ary strong
functor, we can define its \emph{De Morgan dual} $\Orth{\Vcsnot F}$
which is also an $n$-ary strong functor. On objects, we set
$\Strfun{\Orth{\Vcsnot F}}(\Vect Y)=\Orth{\Strfun{\Vcsnot
    F}(\Orth{\Vect Y})}$ and similarly for morphisms. The strength of
$\Orth{\Vcsnot F}$ is defined as the Curry transpose of the following
morphism (remember that
$\Limpl{\Excl X}{\Orth{\Vect Y}}=\Orthp{\Tens{\Excl X}{\Vect Y}}$ up
to canonical iso):%

{\footnotesize%
\begin{center}
  \begin{tikzcd}
    \Excl X\ITens{\Orth{\Strfun{\Vcsnot F}(\Orth{\Vect Y})}}
    \ITens\Strfun{\Vcsnot F}(\Limpl{\Excl X}{\Orth{\Vect Y}})
    \arrow[r,phantom,"\isom"]
    &[-1em]\Excl X
    \ITens\Strfun{\Vcsnot F}(\Limpl{\Excl X}{\Orth{\Vect Y}})
    \ITens{\Orth{\Strfun{\Vcsnot F}(\Orth{\Vect Y})}}
    \arrow[d,"\Strnat{\Vcsnot F}\ITens\Id"]\\[-0.2em]
    \Strfun{\Vcsnot F}(\Orth{\Vect Y})
    \ITens{\Orth{\Strfun{\Vcsnot F}(\Orth{\Vect Y})}}
    \arrow[d,swap,"\Evlin\Compl\Sym"]
    &
    \Strfun{\Vcsnot F}(\Excl X
    \ITens(\Limpl{\Excl X}{\Orth{\Vect Y}}))
    \ITens{\Orth{\Strfun{\Vcsnot F}(\Orth{\Vect Y})}}
    \arrow[l,swap,"\Strfun{\Vcsnot F}(\Evlin)\ITens\Id"]\\[-0.2em]
    \Sbot &
  \end{tikzcd}
\end{center}
}%
\noindent
Then it is possible to prove, using the $\ast$-autonomy of $\LCAT$, that
$\Biorth{\Vcsnot F}$ and $\Vcsnot F$ are canonically isomorphic (as
strong functors)\footnote{In the concrete settings considered in this
  paper, these canonical isos are actually identity maps.}.
As a direct consequence of the definition of
$\Orth{\Vcsnot F}$ and of the canonical iso between
$\Biorth{\Vcsnot F}$ and $\Vcsnot F$ we get:

\begin{lemma}\label{lemma:strfun-comp-orth}
  $\Orthp{\Vcsnot F\Comp(\List{\Vcsnot G}1n)}=\Orth{\Vcsnot
    F}\Comp(\List{\Orth{\Vcsnot G}}1n)$ up to canonical iso.
\end{lemma}

The bifunctor $\mathord\ITens$ can be turned into a strong functor:
one defines the strength as%
\footnote{This definition, as well as the following one, shows that
  our assumption that the strength is available for ``context object''
  of shape $\Excl X$ only cannot be disposed of.}%

{\footnotesize
  \begin{center}
    \begin{tikzcd}
      \Excl X\ITens Y_1\ITens Y_2
      \arrow[r,"\Contr{\Excl X}\ITens\Id"]
      &[2em]\Excl X\ITens \Excl X\ITens Y_1\ITens Y_2
      \arrow[r,phantom,"\isom"]
      &[-1em]\Excl X\ITens Y_1\ITens \Excl X\ITens Y_2
    \end{tikzcd}
  \end{center}
}%
\noindent
By De Morgan duality, this endows $\IPar$ with a strength as well. The
bifunctor $\IPlus$ is also endowed with a strength, simply using the
distributivity of $\ITens$ over $\IPlus$ (which results from the
monoidal closedness of $\LCAT$). By duality again, $\IWith$ inherits a
strength. The functor $\Excl\_$ is equipped with the strength\\
{\footnotesize
  \begin{tikzcd}
    \Tens{\Excl X}{\Excl Y}
    \arrow[r,"\Tens{\Digg X}{\Excl Y}"]
    &[0.6em]\Tens{\Excl{\Excl X}}{\Excl Y}
    \arrow[r,"\Monoidalt"]
    &[-1em]\Exclp{\Tens{{\Excl X}}{Y}}
  \end{tikzcd}
}.

\subsection{Fixed Points of strong functors}
\label{sec:fixpoints-functors}
The following facts are standard in the
literature on fixed points of functors.
\begin{definition}
  Let $\cA$ be a category and $\cF:\cA\to\cA$ be a functor. A
  \emph{coalgebra}%
  \footnote{Not to be confused with the coalgebras of
    Section~\ref{sec:EM-Kl-category} which must satisfy additional
    properties of compatibility with the comonad structure of
    $\Excl\_$.}%
  of $\cF$ is a pair $(A,f)$ where $A$ is an object of $\cA$ and
  $f\in\cA(A,\cF(A))$. Given two coalgebras $(A,f)$ and $(A',f')$ of
  $\cF$, a coalgebra morphism from $(A,f)$ to $(A',f')$ is an
  $h\in\cA(A,A')$ such that $f'\Compl h=\cF(h)\Compl f$.
  The category of coalgebras of the functor $\cF$ will be denoted as
  $\COALGFUN\cA\cF$. The notion of algebra of an endofunctor is
  defined dually (reverse the directions of the arrows $f$ and $f'$)
  and the corresponding category is denoted as $\ALGFUN\cA\cF$.
\end{definition}


By Lambek's Lemma, if $(A,f)$ with $f\in\cA(A,\cF(A))$ is a final
object in $\COALGFUN\cA\cF$ then $f$ is an iso.  We assume that this
iso is always the identity%
\footnote{This assumption is highly debatable from the view point of
  category theory where the notion of equality of objects is not
  really meaningful. It will be dropped in a longer version of this
  paper.}%
as this holds in our concrete models so that this final object
$(\Fungfp\cF,\Id)$ satisfies $\cF(\Fungfp\cF)=\Fungfp\cF$. We focus on
coalgebras rather than algebras for reasons which will become clear
when we deal with fixed points of strong functors.
This universal property of $\Fungfp{\cF}$ gives us a powerful tool for
proving equalities of morphisms.
\begin{lemma}\label{lemma:equations-final-coalgebra}
  Let $A\in\Obj\cA$ and let
  $f_1,f_2\in\cA(A,\Fungfp{\cF})$. If there exists
  $l\in\cA(A,\cF(A))$ such that
  $\cF(f_i)\Compl l=f_i$ for $i=1,2$, then $f_1=f_2$.
\end{lemma}

\begin{lemma}\label{lemma:functor-gfp-general}
  Let $\cF:\cB\times\cA\to\cA$ be a functor such that, for all
  $B\in\Obj\cB$, the category $\COALGFUN \cA{\cF_B}$ has a final
  object. Then there is a functor $\Fungfp\cF:\cB\to\cA$ such that
  $(\Fungfp\cF(B),\Id)$ is the final object of $\COALGFUN \cA{\cF_B}$
  (so that $\cF(B,\Fungfp\cF(B))=\Fungfp\cF(B)$) for each
  $B\in\Obj\cB$, and, for each $g\in\cB(B,B')$, $\Fungfp\cF(g)$ is
  uniquely characterized by $\cF(g,\Fungfp\cF(g))=\Fungfp\cF(g)$.
\end{lemma}

We consider now the same $\Fungfp\cF$ operation applied to strong
functors on a model $\LCAT$ of $\LL$. Let $\Vcsnot F$ be an $n+1$-ary
strong functor on $\LCAT$ (so that $\Strfun{\Vcsnot F}$ is a functor
$\LCAT^n\times\LCAT\to\LCAT$). Assume that for each
$\Vect X\in\Obj{\LCAT^n}$ the category
$\COALGFUN{\LCAT}{\Strfun{\Vcsnot F}_{\Vect X}}$ has a final
object. We have defined $\Fungfp{\Strfun{\Vcsnot F}}:\LCAT^n\to\LCAT$
characterized by
$\Strfun{\Vcsnot F}(\Vect X,\Fungfp{\Strfun{\Vcsnot F}}(\Vect
X))=\Fungfp{\Strfun{\Vcsnot F}}(\Vect X)$ and
$\Strfun{\Vcsnot F}(\Vect f,\Fungfp{\Strfun{\Vcsnot F}}(\Vect
f))=\Fungfp{\Strfun{\Vcsnot F}}(\Vect f)$ for all
$\Vect f\in\LCAT^n(\Vect X,\Vect{X'})$
(Lemma~\ref{lemma:functor-gfp-general}). For each $Y,\Vect X\in\LCAT$,
we define
$\Strnat{\Vcsnot{\Fungfp F}}_{Y,\Vect X}\in\LCAT(\Tens{\Excl
  Y}{\Fungfp{\Strfun{\Vcsnot F}}(\Vect X)},\Fungfp{\Strfun{\Vcsnot
    F}}(\Tens{\Excl Y}{\Vect X}))$. We have%

{\footnotesize%
  \begin{center}
  \begin{tikzcd}
    \Tens{\Excl Y}{\Strfun{\Fungfp{\Vcsnot F}}(\Vect X)} =\Excl
    Y\ITens \Strfun{\Vcsnot F}(\Vect X,\Strfun{\Fungfp{\Vcsnot
        F}}(\Vect X)) \arrow[r,"\Strnat{\Vcsnot F}_{Y,(\Vect X,
      \Strfun{\Fungfp{\Vcsnot F}}(\Vect X))}"] &[2em]\Strfun{\Vcsnot
      F}(\Tens{\Excl Y}{\Vect X}, \Tens{\Excl
      Y}{\Strfun{\Fungfp{\Vcsnot F}}(\Vect X)})
  \end{tikzcd}    
  \end{center}
}%
\noindent%
exhibiting a $\Strfun{\Vcsnot F}_{\Tens{\Excl Y}{\Vect X}}$-coalgebra
structure on $\Tens{\Excl Y}{\Strfun{\Fungfp{\Vcsnot F}}(\Vect
  X)}$. Since $\Strfun{\Fungfp{\Vcsnot F}}(\Tens{\Excl Y}{\Vect X})$ is
the final coalgebra of the functor
$\Strfun{\Vcsnot F}_{\Tens{\Excl Y}{\Vect X}}$, we define
$\Strnat{\Fungfp{\Vcsnot F}}_{Y,\Vect X}$ as the unique morphism
$\Tens{\Excl Y}{\Strfun{\Fungfp{\Vcsnot F}}(\Vect
  X)}\to\Strfun{\Fungfp{\Vcsnot F}}(\Tens{\Excl Y}{\Vect X})$ such that
the following diagram commutes%

  \begin{equation}\label{eq:final-coalg-strength-charact}
    {\footnotesize
      \begin{tikzcd}
      \Tens{\Excl Y}{\Strfun{\Fungfp{\Vcsnot F}}(\Vect X)}
      \arrow[r,dotted,"\Strnat{\Vcsnot F}_{Y,(\Vect X,
        \Strfun{\Fungfp{\Vcsnot F}}(\Vect X))}"]
      \arrow[d,swap,"\Strnat{\Fungfp{\Vcsnot F}}_{Y,\Vect X}"]
      &[0.5em]\Strfun{\Vcsnot F}(\Tens{\Excl Y}{\Vect X},
      \Tens{\Excl Y}{\Strfun{\Fungfp{\Vcsnot F}}(\Vect X)})
      \arrow[d,"{\Strfun{\Vcsnot F}(\Tens{\Excl Y}{\Vect X},
      \Strnat{\Fungfp{\Vcsnot F}}_{Y,\Vect X})}"]
      \\
      \Strfun{\Vcsnot F}(\Tens{\Excl Y}{\Vect X},
      \Strfun{\Fungfp{\Vcsnot F}}(\Tens{\Excl Y}{\Vect X}))
      \arrow[r,equals]
      &\Strfun{\Fungfp{\Vcsnot F}}(\Tens{\Excl Y}{\Vect X})
    \end{tikzcd}
  }
  \end{equation}
\begin{lemma}\label{lemma:strfun-gfp-general}
  Let $\Vcsnot F$ be an $n+1$-ary strong functor on $\LCAT$ such that
  for each $\Vect X\in\Obj{\LCAT^n}$, the category
  $\COALGFUN\LCAT{\Strfun{\Vcsnot F}_{\Vect X}}$ has a final object
  $\Fungfp{\Strfun{\Vcsnot F}_{\Vect X}}$. Then there is a unique
  $n$-ary strong functor $\Fungfp{\Vcsnot F}$ such that
  $\Strfun{\Fungfp{\Vcsnot F}}(\Vect X)=\Fungfp{\Strfun{\Vcsnot
      F}_{\Vect X}}$ (and hence
  $\Strfun{\Vcsnot F}(\Vect X,\Strfun{\Fungfp{\Vcsnot F}}(\Vect
  X))=\Strfun{\Fungfp{\Vcsnot F}}(\Vect X)$),
    $\Strfun{\Vcsnot F}(\Vect f,\Strfun{\Fungfp{\Vcsnot F}}(\Vect
    f))=\Strfun{\Fungfp{\Vcsnot F}}(\Vect f)$ for all
    $\Vect f\in\LCAT^n(\Vect X,\Vect{X'})$
   and
    $\Strfun{\Vcsnot F}(\Tens{\Excl Y}{\Vect
      X},\Strnat{\Fungfp{\Vcsnot F}}_{Y,\Vect X})\Compl\Strnat{\Vcsnot
      F}_{Y,(\Vect X,\Strfun{\Fungfp{\Vcsnot F}}(\Vect
      X))}=\Strnat{\Fungfp{\Vcsnot F}}_{Y,\Vect X}$.
\end{lemma}

\begin{lemma}\label{lemma:strfun-lfp-general}
  Let $\Vcsnot F$ be an $n+1$-ary strong functor on $\LCAT$ such that
  for each $\Vect X\in\Obj{\LCAT^n}$, the category
  $\ALGFUN\LCAT{\Strfun{\Vcsnot F}_{\Vect X}}$ has an initial object
  $\Funlfp{\Strfun{\Vcsnot F}_{\Vect X}}$. Then there is a unique
  $n$-ary strong functor $\Funlfp{\Vcsnot F}$ such that
  $\Strfun{\Funlfp{\Vcsnot F}}(\Vect X)=\Funlfp{\Strfun{\Vcsnot
      F}_{\Vect X}}$ (and hence
  $\Strfun{\Vcsnot F}(\Vect X,\Strfun{\Funlfp{\Vcsnot F}}(\Vect
  X))=\Strfun{\Funlfp{\Vcsnot F}}(\Vect X)$),
    $\Strfun{\Vcsnot F}(\Vect f,\Strfun{\Funlfp{\Vcsnot F}}(\Vect
    f))=\Strfun{\Funlfp{\Vcsnot F}}(\Vect f)$ for all
    $\Vect f\in\LCAT^n(\Vect X,\Vect{X'})$
  and
    $\Strfun{\Vcsnot F}(\Tens{\Excl Y}{\Vect
      X},\Strnat{\Funlfp{\Vcsnot F}}_{Y,\Vect X})\Compl\Strnat{\Vcsnot
      F}_{Y,(\Vect X,\Strfun{\Funlfp{\Vcsnot F}}(\Vect
      X))}=\Strnat{\Funlfp{\Vcsnot F}}_{Y,\Vect X}$.
  Moreover $\Orth{(\Funlfp{\Vcsnot F})}=\Fungfp{(\Orth{\Vcsnot F})}$
\end{lemma}
\begin{proof} Apply Lemma~\ref{lemma:strfun-gfp-general} to the strong
functor $\Orth{\Vcsnot F}$.  \end{proof}

\subsection{A categorical axiomatization of models of $\MULL$}
\label{sec:cat-models}
Our general definition of Seely categorical model of $\MULL$ is
based on the notions and results above.  We refer in particular to
Section~\ref{sec:gen-strong-functors} for the basic definitions of
operations on strong functors in our $\LL$ categorical setting.

\begin{definition}\label{def:categorical-muLL-models}
  A \emph{categorical model} or \emph{Seely model} of $\MULL$ is a
  pair $(\cL,\Vect\cL)$ where
  \begin{Enumerate}
  \item\label{it:def-muLL-model-1} $\cL$ is a Seely
    category\label{enum:seel-mull-1}
  \item\label{it:def-muLL-model-2} $\Vect\cL=(\cL_n)_{n\in\Nat}$ where
    $\cL_n$ is a class of strong functors $\cL^n\to\cL$, and
    $\cL_0=\Obj\cL$\label{enum:seel-mull-2}
  \item\label{it:def-muLL-model-3} if $\Vcstnot X\in\cL_n$ and
    $\Vcstnot X_i\in\cL_k$ (for $i=1,\dots,n$) then
    $\Vcstnot X\Comp\Vect{\Vcstnot X}\in\cL_k$ and all $k$ projection
    strong functors $\cL^k\to\cL$ belong to
    $\cL_k$\label{enum:seel-mull-3}
  \item\label{it:def-muLL-model-4} the strong functors $\ITens$ and
    $\IWith$ belong to $\cL_2$, the strong functor $\Excl\_$ belongs
    to $\cL_1$ and, if $\Vcstnot X\in\cL_n$, then
    $\Orth{\Vcstnot X}\in\cL_n$
    \label{enum:seel-mull-4}
  \item\label{it:def-muLL-model-5} and last, for all
    $\Vcstnot X\in\cL_1$ the category
    $\COALGFUN{\cL}{\Strfun{\Vsnot X}}$ (see
    Section~\ref{sec:fixpoints-functors}) has a final object.  So for
    any $\Vsnot X\in\cL_{k+1}$ there is a strong functor
    $\Fungfp{\Vsnot X}:\cL^k\to\cL$ (see
    Lemma~\ref{lemma:strfun-gfp-general}). It is required that
    $\Fungfp{\Vsnot X}\in\cL_k$.
    \label{enum:seel-mull-5}
  \end{Enumerate}
\end{definition}
\begin{remark}
  By Conditions~\ref{it:def-muLL-model-2}
  and~\ref{it:def-muLL-model-3} (applied with $n=0$), all constant
  strong functors are in $\cL_n$, for all $n$. Therefore given
  $\Vcstnot X\in\cL_{k+1}$ and $\Vect X\in\Obj \cL^k$, the strong
  functor $\Vcsnot X(\_,\Vect X)$ is in $\cL_1$ by
  Condition~\ref{it:def-muLL-model-3}. This explains why we can apply
  Lemma~\ref{lemma:strfun-gfp-general} in
  Condition~\ref{it:def-muLL-model-5}.
\end{remark}
%
Our goal is now to outline the interpretation of $\MULL$ formulas and
proofs in such a model. This requires first to describe the syntax of
formulas and proofs.

\begin{remark}
  One can certainly also define a notion of categorical model of
  $\MULL$ in a linear-non-linear adjunction setting as presented
  in~\cite{Mellies09}. This is postponed to further work.
\end{remark}

\subsubsection{Syntax of $\MULL$}\label{sec:MULL-syntax}
We assume to be given an infinite set of propositional variables
$\LLvars$ (ranged over by Greek letters $\zeta,\xi\dots$). We
introduce a language of propositional $\LL$ formulas with least and
greatest fixed points.
\begin{multline*}
  A,B,\dots \Bnfeq \One \Bnfor \Fbot \Bnfor \Tens AB \Bnfor \Par AB\\
  \Bnfor \Zero \Bnfor \Top \Bnfor \Plus AB \Bnfor \With AB \Bnfor
  \Excl A \Bnfor \Int A \Bnfor\zeta \Bnfor \Lfpll\zeta A \Bnfor
  \Gfpll\zeta A\,.
\end{multline*}
The notion of closed formula is defined as usual, the two last
constructions being the only binders.

\begin{remark}
  In contrast with second-order $\LL$ or dependent type systems
  where open formulas play a crucial role,
  in the case of fixed points, all formulas
  appearing in sequents and other syntactical devices allowing to give
  types to programs will be closed. In our setting, open types/formulas appear
  only locally, for allowing the expression of (least and greatest)
  fixed points.
\end{remark}

We can define two basic operations on formulas.
\begin{Itemize}
\item \emph{Substitution}: $\Subst AB\zeta$, taking care of not
  binding free variables (uses $\alpha$-conversion).
\item \emph{Negation} or \emph{dualization}: defined by induction on
  formulas $\Orth\One=\Fbot$, $\Orth\Fbot=\One$,
  $\Orthp{\Par AB}=\Tens{\Orth A}{\Orth B}$,
  $\Orthp{\Tens AB}=\Par{\Orth A}{\Orth B}$, $\Orth\Zero=\Top$,
  $\Orth\Top=\Zero$, $\Orthp{\With AB}=\Plus{\Orth A}{\Orth B}$,
  $\Orthp{\Plus AB}=\With{\Orth A}{\Orth B}$,
  $\Orthp{\Excl A}=\Int{\Orth A}$, $\Orthp{\Int A}=\Excl{\Orth A}$,
  $\Orth\zeta=\zeta$, $\Orthp{\Lfpll\zeta A}=\Gfpll\zeta{\Orth A}$ and
  $\Orthp{\Gfpll\zeta A}=\Lfpll\zeta{\Orth A}$. Obviously
  $\Biorth A=A$ for any formula $A$.
\end{Itemize}

\begin{remark}
  The only subtle point of this definition is negation of
  propositional variables: $\Orth\zeta=\zeta$. This entails
  $\Orthp{\Subst BA\zeta}=\Subst{\Orth B}{\Orth A}\zeta$ by an easy
  induction on $B$. If we consider $B$ as a compound logical
  connective with placeholders labeled by variables then $\Orth B$ is
  its De~Morgan dual. This definition of $\Orth\zeta$ is also a
  natural way of preventing the introduction of fixed points
  wrt.~variables with negative occurrences.
  As an illustration, if we define as usual $\Limpl AB$ as
  $\Par{\Orth A}{B}$ then we can define
  $E=\Lfpll\zeta{(\With\One{(\Limpl{\Excl\zeta}\zeta}))}$ which
  \emph{looks like} the definition of a model of the pure
  $\lambda$-calculus as a recursive type. But this is only an illusion
  since we actually have
  $E=\Lfpll\zeta{(\With\One{(\Par{\Int\zeta}\zeta}))}$ so that
  $\Limpl{\Excl E}{E}$ \emph{is not a retract} of $E$. And indeed if it were
  possible to define a type $D$ such that $\Limpl{\Excl D}{D}$ is
  isomorphic to (or is a retract of) $D$ then we would be able to type
  all pure $\lambda$-terms in our system and this would contradict the fact
  that $\MULL$ enjoys strong normalization and has a denotational
  semantics based on totality as shown below.
\end{remark}

Our logical system extends the usual unilateral sequent calculus of
classical propositional $\LL$~\cite{Girard87}, see
also~\cite{Mellies09} Section~3.1 and~3.13. In this setting we deal
with sequents $\Seq{\List A1n}$ where the $A_i$'s are formulas. It is
important to notice that the order of formulas in this list is not
relevant, which means that we keep the exchange rule implicit as it is
usual in sequent calculus. To the standard rules%
\footnote{Notice that the \emph{promotion rule} of $\LL$ has a
  condition on contexts similar to that of the rule \Ngfp{} below: to
  deduce $\Seq{\Delta,\Excl A}$ from $\Seq{\Delta,A}$ it is
  required that all formulas in the context $\Delta$ are of shape
  $\Int B$, that is $\Delta=\Int\Gamma$.} %
of~\cite{Mellies09} Fig.~1, we add the two next introduction rules for
fixed point formulas which are essentially borrowed to~\cite{Baelde12}
(see Section~\ref{sec:MULL-comments})

{\footnotesize
\[
\begin{prooftree}
  \hypo{\Seq{\Gamma,\Subst{F}{\Lfpll\zeta F}{\zeta}}}
  \infer1[\Nlfp]{\Seq{\Gamma,\Lfpll\zeta F}}
\end{prooftree}
\quad
\begin{prooftree}
  \hypo{\Seq{\Delta,A}}
  \hypo{\Seq{\Int\Gamma,\Orth A,\Subst{F}{A}{\zeta}}} 
  \infer2[\Ngfp]{\Seq{\Delta,\Int\Gamma,\Gfpll\zeta F}}
\end{prooftree}\,.
\]%
}
By taking, in the last rule, $\Delta=\Orth A$ and proving the left
premise by an axiom, we obtain the following derived rule
\[\begin{prooftree}
  \hypo{\Seq{\Int\Gamma,\Orth A,\Subst{F}{A}{\zeta}}}
  \infer1[\Ngfpbis]{\Seq{\Int\Gamma,\Orth A,\Gfpll\zeta F}}
\end{prooftree}\,.\]
  
\noindent%
The corresponding cut-elimination rule is described in
Section~\ref{sec:MULL-cut-elim}. For the other connectives (which are
the standard connectives of $\LL$), the cut-elimination rules are the
usual ones as described in~\cite{Girard87,Mellies09}.

\subsubsection{Comments}\label{sec:MULL-comments}
Let us summarize and comment the differences between our system and
Baelde's $\MUMALL$.
  \begin{Itemize}
  \item Baelde's logical system is a \emph{predicate calculus} whereas
    our system is a \emph{propositional calculus}. Indeed, Baelde is
    mainly interested in applying $\MUMALL$ to program verification
    where the predicate calculus is essential for expressing
    properties of programs. We have a Curry-Howard perspective where
    formulas are seen at types and proofs as programs and where a
    propositional logical system is sufficient.
  \item Our system has exponentials whereas Balede's system has not
    because they can be encoded in $\MUMALL$ to some extent. However
    the exponentials encoded in that way do not satisfy all required
    isos (in particular the ``Seely morphisms'' are not isos with
    Baelde's exponentials) and this is a serious issue if we want to
    encode some form of $\lambda$-calculus in the system and consider
    it as a programming language.
  \item Our \Ngfp{} rule differs from Baelde's by the fact that we
    admit a context in the right premise. Notice that all formulas of
    this context must bear a $\Int\_$ modality: this restriction is
    absolutely crucial for allowing to express the cut-elimination
    rule in Section~\ref{sec:MULL-cut-elim} which uses an operation of
    substitution of proofs in formulas and this operation uses
    structural rules on the context. The semantic counterpart of this
    operation is described in
    Section~\ref{sec:strong-fun-LL-operations} where it appears
    clearly that it uses the fact that $P$ is an object of
    $\Em\LCAT$. Such a version of \Ngfp{} with a context would be
    problematic in Baelde's system by lack of built-in exponentials.
  \end{Itemize}

\subsubsection{Syntactic functoriality of formulas}%
\label{sec:synt-functoriality}%
The reduction rule for the \Nlfp{}/\Ngfp{} cut requires the
possibility of substituting a proof for a propositional variable in a
formula. More precisely, let $(\zeta,\List\xi 1k)$ be a list of
pairwise distinct propositional variables containing all the free
variables of a formula $F$ and let $\Vect C=(\List C1k)$ be a sequence
of closed formulas. Let $\pi$ be a proof of
$\Seq{\Int\Gamma,\Orth A,B}$, then one defines%
\footnote{Again the fact that the formulas of the context bear a
  $\Int\_$ is absolutely necessary to make this definition possible.} %
a proof $\Substbis F{\pi/\zeta,\Vect C/\Vect\xi}$ of
\[
  \Seq{\Int\Gamma,\Orthp{\Substbis F{A/\zeta,\Vect
        C/\Vect\xi}},\Substbis F{B/\zeta,\Vect C/\Vect\xi}}
\]
by induction on $F$, adapting the corresponding definition
in~\cite{Baelde12}.  We illustrate this definition by two inductive
steps.  Observe, in these examples, how the exchange rule is used
implicitly.

Assume first that $F=\Lfpll\xi G$ (so that $(\zeta,\xi,\List\xi 1k)$
is a list of pairwise distinct variables containing all free variables
of $G$). Let $G'=\Subst G{\Vect C}{\Vect\xi}$ whose only possible free
variables are $\zeta$ and $\xi$.  The proof
$\Substbis F{\pi/\zeta,\Vect C/\Vect\xi}$ is defined by
\[{\footnotesize
    \begin{prooftree}
      \hypo{}
      \ellipsis{$\Substbis {G}{\pi/\zeta,
          \Subst{(\Lfpll\xi {G'})}{B}{\zeta}/\xi,\Vect C/\Vect\xi}$}
      {\Seq{\Int\Gamma,\Orthp{\Substbis {G'}{A/\zeta,
              \Subst{(\Lfpll\xi {G'})}{B}{\zeta}/\xi}},
          \Substbis {G'}{B/\zeta,\Subst{(\Lfpll\xi {G'})}{B}{\zeta}/\xi}}}
      \infer1[\Nlfp]{\Seq{\Int\Gamma,\Orthp{\Substbis {G'}{A/\zeta,
              \Subst{(\Lfpll\xi {G'})}{B}{\zeta}/\xi}},
          \Subst{(\Lfpll\xi {G'})}{B}{\zeta}}}
      \infer1[\Ngfpbis]{\Seq{\Int\Gamma,
          \Orthp{\Subst{(\Lfpll\xi {G'})}{A}{\zeta}},
          \Subst{(\Lfpll\xi {G'})}{B}{\zeta}}}
    \end{prooftree}
}
\]
Notice that this case uses the additional parameters $\Vect C$ in the
definition of this substitution with $k+1$ parameters in the inductive
hypothesis. To see that the last inference in this deduction is an
instance of \Ngfpbis{}, set $H=\Subst{\Orth{G'}}{\Orth A}\zeta$ and
notice that %
  $\Orthp{\Substbis {G'}{A/\zeta, \Subst{(\Lfpll\xi
        {G'})}{B}{\zeta}/\xi}}= \Substbis{H}{ \Orth{\Subst{(\Lfpll\xi
        {G'})}{B}{\zeta}}/\xi}$ 
and $\Orthp{\Subst{(\Lfpll\xi {G'})}{A}{\zeta}}=\Gfpll\xi{H}$.
Another example is $F=\Tens{G_1}{G_2}$:
$\Substbis F{\pi/\zeta,\Vect C/\Vect\xi}$ is defined as
\begin{center}{\footnotesize
 \begin{prooftree}
   \hypo{}
   \ellipsis{$\Substbis{G_1}{\pi/\zeta,\Vect C/\Vect\xi}$}
   {\Seq{\Int\Gamma,\Orthp{{\Subst{G_1'}{A}{\zeta}}},{\Subst{G_1'}{B}{\zeta}}}}
   \hypo{}
   \ellipsis{$\Substbis{G_2}{\pi/\zeta,\Vect C/\Vect\xi}$}
   {\Seq{\Int\Gamma,\Orthp{{\Subst{G_2'}{A}{\zeta}}},{\Subst{G_2'}{B}{\zeta}}}}
   \infer2[\Ntens]
   {\Seq{\Int\Gamma,\Int\Gamma,\Orthp{{\Subst{G_1'}{A}{\zeta}}},
       \Orthp{{\Subst{G_2'}{A}{\zeta}}},
       \Tens{\Subst{G_1'}{B}{\zeta}}{\Subst{G_2'}{B}{\zeta}}}}
   \infer[double]1[\Ncontr]{\Seq{\Int\Gamma,\Orthp{{\Subst{G_1'}{A}{\zeta}}},
       \Orthp{{\Subst{G_2'}{A}{\zeta}}},
       \Tens{\Subst{G_1'}{B}{\zeta}}{\Subst{G_2'}{B}{\zeta}}}}
   \infer1[\Npar]{\Seq{\Int\Gamma,
       \Orthp{{\Subst{G_1'}{A}{\zeta}}}\IPar\Orthp{{\Subst{G_2'}{A}{\zeta}}},
       \Tens{\Subst{G_1'}{B}{\zeta}}{\Subst{G_2'}{B}{\zeta}}}}
 \end{prooftree}
}
\end{center}
Observe that we use in an essential way the fact that all formulas of
the context are of shape $\Int H$ (even if $F$ is exponential-free)
when we apply contraction rules on this context.

\subsubsection{Cut elimination}%
\label{sec:MULL-cut-elim}%
The only reduction that we will mention here is \Nlfp{}/\Ngfp{}.
Let $\theta$ be
\[{\footnotesize
 \begin{prooftree}
   \hypo{}
   \ellipsis{$\pi$}{\Seq{\Lambda,\Subst F{\Lfpll\zeta F}\zeta}}
   \infer1[\Nlfp]{\Seq{\Lambda,\Lfpll\zeta F}}
   \hypo{}
   \ellipsis{$\lambda$}{\Seq{\Delta,\Orth A}}
   \hypo{}
   \ellipsis{$\rho$}{\Seq{\Int\Gamma,A,\Orthp{\Subst FA\zeta}}}
   \infer2[\Ngfp]{\Seq{\Delta,\Int\Gamma,\Orthp{\Lfpll\zeta F}}}
   \infer2[\Ncut]{\Seq{\Lambda,\Delta,\Int\Gamma}}
 \end{prooftree}
}
\]
and let $\rho'$ be the proof
\[{\footnotesize
  \begin{prooftree}
    \hypo{}
    \ellipsis{$\rho$}{\Seq{\Int\Gamma,A,\Orthp{\Subst FA\zeta}}}
    \infer1[\Ngfpbis]{\Seq{\Int\Gamma,A,\Orthp{\Lfpll\zeta F}}}
  \end{prooftree}}
\]
Then $\theta$ reduces to the proof shown in Figure~\ref{fig:fp-reduction}.
\begin{figure*}
  \centering
  {\footnotesize %
    \begin{prooftree}
      \hypo{} \ellipsis{$\Subst F{\rho'}\zeta$}
      {\Seq{\Int\Gamma,\Subst FA\zeta},\Orthp{\Subst F{\Lfpll\zeta
            F}\zeta}} \hypo{} \ellipsis{$\pi$}{\Seq{\Lambda,\Subst
          F{\Lfpll\zeta F}\zeta}}
      \infer2[\Ncut]{\Seq{\Lambda,\Int\Gamma},\Subst FA\zeta} \hypo{}
      \ellipsis{$\lambda$}{\Seq{\Delta,\Orth A}} \hypo{}
      \ellipsis{$\rho$}{\Seq{\Int\Gamma,A,\Orthp{\Subst FA\zeta}}}
      \infer2[\Ncut]{\Seq{\Delta,\Int\Gamma,\Orthp{\Subst FA\zeta}}}
      \infer2[\Ncut]{\Seq{\Lambda,\Delta,\Int\Gamma,\Int\Gamma}}
      \infer[double]1[\Ncontr]{\Seq{\Lambda,\Delta,\Int\Gamma}}
    \end{prooftree} %
  }
  \caption{Reduction \Nlfp{}/\Ngfp}
  \label{fig:fp-reduction}
\end{figure*}
This reduction rule uses the functoriality of formulas as well as the
$\wn$-contexts in the \Ngfp{} rule.

\begin{remark}
  In~\cite{Baelde12} it is shown that $\MUMALL$ enjoys
  cut-elimination. We will show in a further paper how this method
  based on reducibility can be adapted to our $\MULL$. Notice that a
  cut-free proof has not the sub-formula property in general because
  of rule \Ngfp{}. Though, the normalization theorem makes sure that a
  proof of a sequent \emph{which does not contain any $\nu$-formula}
  reduces to a cut-free proof enjoying the sub-formula property.
\end{remark}

\subsubsection{Interpreting formulas and proofs (outline)}
%
We assume to be given a $\MULL$ Seely model $(\cL,\Vect\cL)$, see
Section~\ref{sec:cat-models}. With any formula $A$ and any
repetition-free sequence $\Vect\zeta=(\List\zeta 1k)$ of type
variables containing all the free variables of $A$, we associate
$\Tsem A_{\Vect\zeta}\in\cL_k$ in the obvious way, for instance
$\Tsem{\Tens AB}_{\Vect\zeta}=\mathord\ITens\Comp(\Tsem
A_{\Vect\zeta},\Tsem B_{\Vect\zeta})\in\cL_k$ by
conditions~(\ref{enum:seel-mull-4}) and~(\ref{enum:seel-mull-3}) in
Definition~\ref{def:categorical-muLL-models} and
$\Tsem{\Gfpll\zeta A}_{\Vect\zeta}=\Fungfp{(\Tsem
  A_{\Vect\zeta,\zeta})}$ using
condition~(\ref{enum:seel-mull-5}). Then
$\Tsem{\Orth A}_{\Vect\zeta}=\Orth{\Tsem{A}_{\Vect\zeta}}$ up to a
natural isomorphism. In this outline, we keep symmetric monoidality
isomorphisms of $\cL$ and of $\Excl\_$ implicit (see for
instance~\cite{Ehrhard18} how \emph{monoidal trees} allow to take them
into account).
With any $\Gamma=(\List A1n)$ we associate an object $\Tsem\Gamma$ of
$\cL$ and with any proof $\pi$ of $\Seq\Gamma$ we associate a morphism
$\Psem\pi\in\cL(\One,\Tsem\Gamma)$ using the categorical constructs of
$\cL$ in a straightforward way, see~\cite{Mellies09}.
Then one proves that if $\pi$ and $\pi'$ are proofs of $\Seq\Gamma$
and $\pi$ reduces to $\pi'$ by the cut-elimination rules, then
$\Psem\pi=\Psem{\pi'}$. This is done by an inspection of the various
cut-elimination rules. 
%
%
In the case of \Nlfp{}/\Ngfp{} cut-elimination, we need the following
lemma.
\begin{lemma}
  Let $\Gamma=(\List D1n)$ be a sequence of closed formulas and let
  $P=\Excl{\Psem{D_1}}\ITens\cdots\ITens\Excl{\Tsem{D_n}}$, considered
  as an object of $\Em\LCAT$. Let $F$ be a formula and
  $\zeta,\List\xi 1k$ be a repetition-free list of variables
  containing all the free variables of $F$ so that
  $\Vcsnot F=\Tsem F_{\zeta,\Vect\xi}$ is a strong functor
  $\LCAT^{k+1}\to\LCAT$ which belongs to $\LCAT_k$.  As shown in
  Section~\ref{sec:gen-strong-functors} this strong functor lifts to a
  functor $\Mfun{\Vcsnot F}P:\Kcomod\LCAT P^{k+1}\to \Kcomod\LCAT P$.
  Let $\pi$ be a proof of $\Seq{\Int\Gamma,\Orth A,B}$, so that
  $\Psem\pi\in\Kcomod\LCAT P(\Tsem A,\Tsem B)$.  Let
  $\Vect C=(\List C1k)$ be a list of closed formulas. Then
  \begin{multline*}
    \Psem{\Substbis F{\pi/\zeta,\Vect C/\Vect \xi}}
    =\Mfun{\Vcsnot F}P(\Psem\pi,\Tsem{C_1},\dots,\Tsem{C_k})\\
    \in\Kcomod\LCAT P(
      \Strfun{\Vcsnot F}(\Tsem A,\Tsem{C_1},\dots,\Tsem{C_k}),
      \Strfun{\Vcsnot F}(\Tsem B,\Tsem{C_1},\dots,\Tsem{C_k}))\,.
  \end{multline*}
\end{lemma}
\noindent%
The proof of the lemma is a simple verification. Notice that we use
the fact that the objects of $\Kcomod\LCAT P$ are the same as those of
$\LCAT$.

\section{Sets and relations}\label{sec:sets-rel}

The category $\REL$ has sets as objects, and given sets $E$ and $F$,
$\REL(E,F)=\Part{E\times F}$. Identity is the diagonal relation and
composition is the usual composition of relations, denoted by simple
juxtaposition. If $t\in\REL(E,F)$ and $u\subseteq E$ then
$\Matappa tu=\Eset{b\in F\St\exists a\in u\ (a,b)\in t}$.

\subsection{$\REL$ as a model of $\LL$.}\label{sec:REL-LL-model}
This category is a well-known model of $\LL$ in which
$\One=\Botlin=\Eset{\Onelem}$,
$\Tens EF=(\Limpl EF)=\Par EF=E\times F$ so that $\Orth E=E$. As to
the additives, $\Zero=\Top=\emptyset$ and
$\Bwith_{i\in I}E_i=\Bplus_{i\in I}E_i=\Union_{i\in I}\Eset i\times
E_i$. The exponentials are given by $\Excl E=\Int E=\Mfin E$ (finite
multisets of elements of $E$ which are functions $m:E\to\Nat$ such
that $m(a)\not=0$ for finitely many $a$'s; addition of multisets is
defined in the obvious pointwise way, and $\Mset{\List a1k}$ is the
multiset which maps $a$ to the number of $i$'s such that $a_i=a$).
For the additives and multiplicatives, the operations on morphisms are
defined in the obvious way. Let us be more specific about the
exponentials. Given $s\in\REL(E,F)$, $\Excl s\in\REL(\Excl E,\Excl F)$
is
\(\Excl s=\{(\Mset{\List a1n},\Mset{\List b1n})\St n\in\Nat \text{ and
}\forall i\ (a_i,b_i)\in s\}\), $\Der E\in\REL(\Excl E,E)$ is given by
$\Der E=\Eset{(\Mset a,a)\St a\in E}$ and
$\Digg E\in\REL(\Excl E,\Excll E)$ is given by
$\Digg E=\{(m_1+\cdots+m_n,\Mset{\List m1n})\St\forall i\ m_i\in\Mfin
E\}$. Last $\Seelyz\in\REL(\One,\Excl\Top)$ is
$\Seelyz=\Eset{(\Onelem,\Mset{})}$ and
$\Seelyt_{E,F}\in\REL(\Tens{\Excl E}{\Excl F},\Excl{(\With EF)})$ is
given by
\begin{multline*}
  \Seelyt_{E,F}=\{((\Mset{\List a1k},\Mset{\List
    b1l}),\\
\Mset{(1,a_1),\dots,(1,a_k),(2,b_1),\dots,(2,b_l)})
  \St\\
\List a1k\in E\text{ and }\List b1l\in F\}\,.
\end{multline*}
Weakening $\Weak E\in\REL(\Excl E,\One)$ and
$\Contr E\in\REL(\Excl E,\Tens{\Excl E}{\Excl E})$ are given by
$\Weak E=\Eset{(\Mset{},\Onelem)}$ and
\(\Contr E=\{(m_1+m_2,(m_1,m_2))\St m_i\in\Mfin E\text{ for
  }i=1,2\}\).

\subsection{Locally continuous functors on $\REL$}
\label{sec:hom-continuous}


The following considerations on continuity of functors are standard,
see~\cite{Wand79}.  A functor $\Vsnot F:\REL^n\to\REL$ is locally
continuous if, for all $\Vect E,\Vect F\in\REL^n$ and all directed set
$D\subseteq\REL^n(\Vect E,\Vect F)$, one has
$\Vsnot F(\Union D)=\Union\Eset{\Vsnot F(\Vect s)\St\Vect s\in D}$ (we
use $\Union D$ for the component-wise union). This implies in
particular that if $\Vect s\subseteq\Vect t$, one has
$\Vsnot F(\Vect s)\subseteq\Vsnot F(\Vect t)$ (taking
$D=\Eset{\Vect s,\Vect t}$).
To simplify notations assume that $n=1$ (but what follows holds for
all values of $n$).

\begin{lemma}\label{lemma:rel-embedding-retraction}
  Let $E$ and $F$ be sets and let $s\in\REL(E,F)$ and
  $t\in\REL(F,E)$. Assume that $t\Compl s=\Id_E$ and that
  $s\Compl t\subseteq\Id_F$. Then $s$ is (the graph of) an injective
  function and $t=\{(b,a)\in F\times E\St (a,b)\in s\}$.
\end{lemma}

\begin{lemma}
  Let $\Vsnot F:\REL\to\REL$ be a locally continuous functor.  Assume
  that $E\subseteq F$ and let
  $\Relii_{E,F}=\Eset{(a,a)\St a\in E}\in\REL(E,F)$ and
  $\Relip_{E,F}=\Eset{(a,a)\St a\in E}\in\REL(F,E)$. Then
  $\Vsnot F(\Relii_{E,F})\in\REL(\Vsnot F(E),\Vsnot F(F))$ is an
  injective function.
\end{lemma}
\begin{proof}
  We have $\Relip_{E,F}\Compl\Relii_{E,F}=\Id_E$ and
  $\Relii_{E,F}\Compl\Relip_{E,F}\subseteq\Id_F$ and hence
  $\Vsnot F(\Relip_{E,F})\Compl\Vsnot F(\Relii_{E,F})=\Id$ by
  functoriality and
  $\Vsnot F(\Relii_{E,F})\Compl\Vsnot F(\Relip_{E,F})\subseteq\Id$ by
  local continuity. The announced property results from
  Lemma~\ref{lemma:rel-embedding-retraction}.
\end{proof}

Let $\RELI$ be the category whose objects are sets and morphisms are
set inclusions (so that $\RELI(E,F)$ has $\Relii_{E,F}$ as unique
element if $E\subseteq F$ and is empty otherwise). Then $\Relii$ can
be thought of as the ``inclusion functor'' $\RELI\to\REL$, acting as
the identity on objects. Obviously, $\RELI$ is
cocomplete\footnote{Notice that it is not complete, for instance is
  has no final object.}.


\begin{proposition}\label{prop:hom-continuous-dir-cocontinuous}
  If $\Vsnot F:\REL\to\REL$ is locally continuous then
  $\Vsnot F\Compl\Relii:\RELI\to\REL$ is directed-cocontinuous (that
  is, preserves the colimits of directed sets of sets).
\end{proposition}
The proof can be found in \cite{Wand79}.
We know that a locally continuous functor $\Vsnot F$ maps inclusions to
injections, we shall say that $\Vsnot F$ is \emph{strict} if it maps
inclusions to inclusions, that is, if $E\subseteq F$ then
$\Vsnot F(E)\subseteq\Vsnot F(F)$ and
$\Vsnot F(\Relii_{E,F})=\Relii_{\Vsnot F(E),\Vsnot F(F)}$ (which
implies $\Vsnot F(\Relip_{E,F})=\Relip_{\Vsnot F(E),\Vsnot F(F)}$). As
a direct consequence of
Proposition~\ref{prop:hom-continuous-dir-cocontinuous}, we get:

\begin{lemma}\label{lemma:str-hom-cont-scott}
  If $\Vsnot F$ is strict locally continuous then, for any directed
  set of sets $\cD$, one has
  $\Vsnot F(\Union\cD)=\Union_{E\in\cD}\Vsnot F(E)$.
\end{lemma}

\subsection{Variable sets and basic constructions on
  them}\label{sec:variable-sets}

\begin{definition}
  An $n$-ary \emph{variable set} is a strong functor $\Vsnot V:\REL^n\to\REL$
  such that $\Strfun{\Vsnot V}$ is locally continuous and strict.
\end{definition}
\label{sec:basic-variable-sets}
By the general considerations of
Section~\ref{sec:gen-strong-functors}, there is a constant strong
functor $\REL^n\to\REL$ with value $E$ for each set $E$. There are
projection strong functors $\REL^n\to\REL$, $\times$ (that is
$\ITens$) and $+$ (that is $\IPlus$) define strong functors
$\REL^2\to\REL$, $\Mfin\_$ (that is $\Excl\_$) defines a strong
functor $\REL\to\REL$. Strong functors on $\REL$ are stable under
composition, and if $\Vsnot V$ is a strong functor $\REL^n\to\REL$
then there is a ``dual'' strong functor $\Orth{\Vsnot V}$ (which is
actually identical to $\Vsnot V$ in this very simple model). To check
that these strong functors $\Vsnot V$ are variable sets we have only
to check that the underlying functors $\Strfun{\Vsnot V}$ are strict
locally continuous.

We deal with $\Excl\_$ and composition, the other
cases are similar.
We already defined the functor%
\footnote{We use the same notation for the strong functor and its
  underlying functor, this slight abuse of notations should not lead
  to confusions.} %
$\Excl\_$ in Section~\ref{sec:REL-LL-model}.
If $s\subseteq t\in\REL(E,F)$, it follows from the definition
that $\Excl s\subseteq\Excl t$.  Let $D\subseteq\REL(E,F)$ be
directed, we prove
$\Excl{(\Union D)}\subseteq\Union_{s\in D}\Excl s$: an element of
$\Excl{(\Union D)}$ is a pair
$(\Mset{\List a1k},\Mset{\List b1k})$
with $(a_i,b_i)\in\Union D$ for $i=1,\dots,k$. Since $D$ is directed,
there is an $s\in D$ such that $(a_i,b_i)\in s$ for $i=1,\dots,k$ and
the inclusion follows. Strictness is obvious.

\label{sec:comp-var-sets}

Let $\Vsnot V_i:\REL^n\to\REL$ be variable sets for $i=1,\dots,k$ and
let $\Vsnot W:\REL^k\to\REL$ be a variable set. Then the functor
$\Strfun{\Vsnot W}\Comp\Vect{\Strfun{\Vsnot V}}:\REL^n\to\REL$ is
clearly strict locally continuous (since these conditions
are preservation properties) from which it follows that the strong
functor $\Vsnot U=\Vsnot W\Comp\Vect{\Vsnot V}$ is a variable
type.

\subsubsection{Fixed point of a variable set}\label{sec:VS-fixpoints}
Let $\Vsnot F:\REL\to\REL$ be a strict locally continuous functor. Since
$\emptyset\subseteq\Vsnot F(\emptyset)$ we have
$\Vsnot F^n(\emptyset)\subseteq\Vsnot F^{n+1}(\emptyset)$ for all
$n\in\Nat$, by induction on $n$ and hence
$\Vsnot F(\Union_{n=0}^\infty\Vsnot F^n(\emptyset))=\Union\Vsnot
F^n(\emptyset)$ by Lemma~\ref{lemma:str-hom-cont-scott} since
$\Eset{\Vsnot F^n(\emptyset)\St n\in\Nat}$ is directed. Let
$\Funfp{\Vsnot F}=\Union_{n=0}^\infty\Vsnot F^n(\emptyset)$, so that
$(\Funfp{\Vsnot F},\Id_{\Funfp{\Vsnot F}})$ is an
$\Vsnot F$-coalgebra.

\begin{lemma}\label{lemma:rel-fixpoint-final}
  The coalgebra $(\Funfp{\Vsnot F},\Id)$ is final in $\COALGFUN\REL{\Vsnot F}$.
\end{lemma}
Notice that $(\Funfp{\Vsnot F},\Id)$ is also an initial object in
$\ALGFUN\REL{\Vsnot F}$. When we insist on considering
$\Funfp{\Vsnot F}$ as a final coalgebra, we denote it as
$\Fungfp{\Vsnot F}$.

\begin{lemma}\label{lemma:hom-conts-stable-fixpoint}
  Let $\Vsnot F:\REL^{n+1}\to\REL$ be a strict locally continuous
  functor. The functor $\Fungfp{\Vsnot F}:\REL^n\to\REL$ is strict
  locally continuous.
\end{lemma}

Let $\Vsnot V:\REL^{n+1}\to\REL$ be a variable set, by
Lemma~\ref{lemma:strfun-gfp-general}, there is a unique strong functor
$\Fungfp{\Vsnot V}:\REL^n\to\REL$ which is characterized by:
$\Strfun{\Fungfp{\Vsnot V}}(\Vect E)=\Fungfp{\Strfun{\Vsnot V}_{\Vect
    E}}$, for each $\Vect s\in\REL^n(\Vect E,\Vect F)$,
$\Strfun{\Fungfp{\Vsnot V}}(\Vect s)=\Strfun{\Vsnot V}(\Vect
s,\Fungfp{\Strfun{\Vsnot V}}(\Vect s))$ and last
$\Strfun{\Vsnot V}(\Excl E\ITens\Vect F,\Strnat{\Vsnot V}_{E,\Vect
  F})=\Strnat{\Vsnot V}_{E,\Vect F}$.

\begin{lemma}\label{lemma:REL-variable-set-fixpoint}
  The functor $\Fungfp{\Vsnot V}$ is a variable set.
\end{lemma}
\begin{proof}
  By the conditions above satisfied by $\Fungfp{\Vsnot V}$ we have
  that $\Strfun{\Fungfp{\Vsnot V}}=\Fungfp{\Strfun{\Vsnot V}}$ and
  hence $\Strfun{\Fungfp{\Vsnot V}}$ is strict locally continuous by
  Lemma~\ref{lemma:hom-conts-stable-fixpoint}.
\end{proof}

\subsubsection{A model of $\MULL$ based on variable sets}
\label{sec:strong-VS-Seely-model}
Let $\VREL n$ be the class of all $n$-ary variable sets, so that
$\VREL 0=\Obj\REL$. The fact that $(\REL,(\VREL n)_{n\in\Nat})$ is a
Seely model of $\MULL$ in the sense of
Definition~\ref{def:categorical-muLL-models} results mainly from the
fact that we take \emph{all} variable sets in the $\VREL n$'s so that
there is essentially nothing to check. More explicitly:
(\ref{enum:seel-mull-1}) holds by Section~\ref{sec:REL-LL-model},
(\ref{enum:seel-mull-2}) holds by construction,
(\ref{enum:seel-mull-3}) holds by the fact that variable sets compose
as explained in Section~\ref{sec:comp-var-sets} (notice that this
condition refers to the general composition of strong functors defined
at the beginning of Section~\ref{sec:strong-fun-LL-operations}),
(\ref{enum:seel-mull-4}) holds by
Section~\ref{sec:basic-variable-sets} and by the fact that the
De~Morgan dual of a strong functor is strong, see Section~\ref
{sec:strong-fun-LL-operations} and
(\ref{enum:seel-mull-5}) holds by Lemma~\ref{lemma:REL-variable-set-fixpoint}.

\section{Non-uniform totality spaces}\label{sec:NUTS}%
We enrich the model of Section~\ref{sec:sets-rel} with a notion of
totality, we use notations from that section for operations on sets and
relations.

\subsection{Basic definitions.}
Let $E$ be a set and let $\cT\subseteq\Part E$. We define
$
  \Orth\cT=\{u'\subseteq E\St\forall u\in\cT\ u\cap u'\not=\emptyset\}
$.
If $\cS\subseteq\cT\subseteq\Part E$ then
$\Orth\cT\subseteq\Orth\cS$. We also have $\cT\subseteq\Biorth\cT$ and
therefore $\Triorth\cT=\Orth\cT$.  The biorthogonal closure has a nice and
simple characterization.
\begin{lemma}\label{lemma:NUTS-biorth-uppper-closed}
  Let $\cT\subseteq\Part E$, then
  $\Biorth\cT=\Upcl\cT=\{v\subseteq E\St\exists u\in\cT\ u\subseteq
    v\}$.
\end{lemma}
\begin{proof}
  The $\supseteq$ direction is obvious, the converse is not difficult
  either: let $u\in\Biorth\cT$.  Then $E\setminus u\notin\Orth\cT$, so
  there is $v\in\cT$ such that $v\cap(E\setminus u)=\emptyset$, that
  is $v\subseteq u$. Hence $u\in\Upcl\cT$.
\end{proof}

A \emph{non-uniform totality space} (NUTS) is a pair
$X=(\Web X,\Total X)$ where $\Web X$ is a set and
$\Total X\subseteq\Part{\Web X}$ satisfies
$\Total X=\Biorth{\Total X}$, that is $\Total X=\Upcl{\Total X}$. Of
course we set $\Orth X=(\Web X,\Orth{\Total X})$.

\begin{example}
  Let $X=(\Nat,\Total X)$ where $\Total X$ is the set of all infinite
  subsets of $\Nat$. It is a NUTS because a superset of an infinite
  set is infinite. Then $\Web{\Orth X}=\Nat$ and $\Total{\Orth X}$ is
  the set of all cofinite subsets of $\Nat$ (the subsets $u$ of $\Nat$
  such that $\Nat\setminus u$ is finite).
  If, with the same web $\Nat$, we take
  $\Total X=\Eset{u\subseteq\Nat\St u\not=\emptyset}$ (again
  $\Total X=\Upcl{\Total X}$ obviously), then
  $\Total{\Orth X}=\Eset\Nat$.
\end{example}

We define four basic NUTS:
$\Zero=(\emptyset,\emptyset)$,
$\Top=(\emptyset,\Eset\emptyset)$ and
$\One=\Botlin=(\Eset\Oneelem,\Eset{\Eset\Oneelem})$.
%
Given NUTS $X_1$ and $X_2$ we define a NUTS $\Tens{X_1}{X_2}$ by
$\Web{\Tens{X_1}{X_2}}=\Web{X_1}\times\Web{X_2}$ and
\(
  \Total{\Tens{X_1}{X_2}}
  =\Upcl{\Eset{\Tens{u_1}{u_2}\St u_i\in\Total{X_i}\text{ for }i=1,2}}
  \) where $\Tens{u_1}{u_2}=u_1\times u_2$.
And then we define $\Limpl XY=\Orthp{\Tens X{\Orth Y}}$.

\begin{lemma}\label{lemma:nuts-limpl-charact}
  $t\in\Total{\Limpl XY}\Equiv\forall u\in\Total X\ \Matappa tu\in\Total Y$.
\end{lemma}

We define the category $\NUTS$ whose objects are the NUTS and
$\NUTS(X,Y)=\Total{\Limpl XY}$, composition being defined as the usual
composition in $\REL$ (relational composition) and identities as the
diagonal relations. Lemma~\ref{lemma:nuts-limpl-charact} shows that we
have indeed defined a category.

\subsubsection{Multiplicative structure}

\begin{lemma}\label{lemma:NUTS-iso-charact}
  Let $X$ and $Y$ be NUTS and $t\in\NUTS(X,Y)$. Then $t$ is an iso in
  $\NUTS$ iff $t$ is (the graph of) a bijection $\Web X\to\Web Y$ such
  that $\forall u\subseteq\Web X\ u\in\Total X\Equiv t(u)\in\Total Y$.
\end{lemma}

\begin{lemma}\label{lemma:nuts-orth-morph}
  Let $t\subseteq\Web X\times\Web Y$. One has $t\in\NUTS(X,Y)$ iff
  $\Orth t=\Eset{(b,a)\St (a,b)\in t}\in\NUTS(\Orth Y,\Orth X)$.
\end{lemma}
\begin{proof}
  This is an obvious consequence of
  Lemma~\ref{lemma:nuts-limpl-charact} and of the fact that
  $\Limplp XY=\Orth{\Tensp X{\Orth Y}}$ and
  $\Limplp{\Orth Y}{\Orth X}=\Orthp{\Tens{\Orth Y}{X}}$.
\end{proof}

\begin{lemma}\label{lemma:limpl-tens-charact}
  Let $t\subseteq\Web{\Limpl{\Tens{X_1}{X_2}}{Y}}$. One has
  $t\in\NUTS(\Tens{X_1}{X_2},Y)$ iff for all $u_1\in\Total{X_1}$ and
  $u_2\in\Total{X_2}$ one has
  $\Matappa t{\Tensp{u_1}{u_2}}\in\Total Y$.
\end{lemma}

\begin{lemma}\label{lemma:Assoc-tens-limpl}
  The bijection $\Assoc_{\Web{X_1},\Web{X_2},\Web Y}$ is an isomorphism from
  $\Limpl{\Tensp{X_1}{X_2}}{Y}$ to $\Limpl{X_1}{\Limplp{X_2}{Y}}$.
\end{lemma}

We turn now $\ITens$ into a functor, its action on morphisms
being defined as in $\REL$. Let $t_i\in\NUTS(X_i,Y_i)$ for $i=1,2$, we
have $\Tens{t_1}{t_2}\in\NUTS(\Tens{X_1}{X_2},\Tens{Y_1}{Y_2})$ by
Lemma~\ref{lemma:limpl-tens-charact} and by the equation
\(
  \Matappa{\Tensp{t_1}{t_2}}{\Tensp{u_1}{u_2}}
  =\Tens{\Matappap{t_1}{u_1}}{\Matappap{t_2}{u_2}}
\).
This functor is monoidal, with unit $\One$ and symmetric monoidality
isomorphisms $\Leftu$, $\Rightu$, $\Sym$ and $\Assoc$ defined as in
$\REL$. The only non-trivial thing to check is that $\Assoc$ is indeed
a morphism, namely
\(
  \Assoc_{\Web{X_1},\Web{X_2},\Web{X_3}}
  \in\NUTS(\Tens{\Tensp{X_1}{X_2}}{X_3},\Tens{X_1}{\Tensp{X_2}{X_3}})
\).
This results from Lemma~\ref{lemma:Assoc-tens-limpl} and from the
observation that
\(
  \Orthp{\Tens{\Tensp{X_1}{X_2}}{X_3}}
  =\Limplp{\Tensp{X_1}{X_2}}{\Orth{X_3}}\) and\\
  \(\Orthp{\Tens{X_1}{\Tensp{X_2}{X_3}}}
  =\Limplp{X_1}{\Limplp{X_2}{\Orth{X_3}}}
\).

The SMC category $\NUTS$ is closed, with $\Limpl XY$ as internal hom
object from $X$ to $Y$, and evaluation morphism
\( \Evlin=\Eset{(((a,b),a),b\St a\in\Web X\text{ and }b\in\Web Y} \)
which indeed belongs to $\NUTS(\Tens{\Limplp XY}{X},Y)$ by
Lemma~\ref{lemma:limpl-tens-charact} since, for all
$t\in\Total{\Limpl XY}$ and $u\in\Total X$ we have
\( \Matapp\Evlin{\Tensp tu}=\Matapp tu\in\Total Y \).  This category
$\NUTS$ is also $\ast$-autonomous with dualizing object $\Fbot=\One$.

\subsubsection{Additive structure}

Let $(X_i)_{i\in I}$ be an at most countable family of objects of
$\NUTS$. We define $X=\Bwith_{i\in I}X_i$ by:
$\Web X=\Union_{i\in I}\Eset i\times\Web{X_i}$ and
\( \Total X=\{u\subseteq\Web X\St\forall i\in I\ \Matappa{\Proj
  i}{u}\in\Total{X_i}\} \)
where $\Proj i=\Eset{((i,a),a)\St a\in \Web{X_i}}$.  It is clear that
$\Total X=\Upcl{\Total X}$ and hence $X$ is an object of $\NUTS$. By
definition of $X$ and by Lemma~\ref{lemma:nuts-limpl-charact} we have
$\forall i\in I\ \Proj i\in\NUTS(X,X_i)$. Given
$\Vect t=(t_i)_{i\in I}$ with $\forall i\in I\ t_i\in\NUTS(Y,X_i)$, we
have $\Tuple{\Vect t}\in\NUTS(Y,X)$ as easily checked (using
Lemma~\ref{lemma:nuts-limpl-charact} again). It follows that
$(\Bwith_{i\in I}X_i,(\Proj i)_{i\in I})$ is the cartesian product of
the $X_i$'s in $\NUTS$.  This shows that the category $\NUTS$ has all
countable products and hence is cartesian.  Since it is
$\ast$-autonomous, the category $\NUTS$ is also cocartesian, coproduct
being given by $\Bplus_{i\in I}X_i=\Orthp{\Bwith_{i\in
    I}\Orth{X_i}}$. It follows that
$X=\Bplus_{i\in I}X_i=\Orthp{\Bwith_{i\in I}\Orth{X_i}}$
satisfies 
$\Web X=\Union_{i\in I}\Eset i\times\Web{X_i}$
and %

{\footnotesize
\[\Total X=
   \{v\subseteq\Union_{i\in I}\Eset i\times\Web{X_i}
   \St\exists i\in I\,\exists u\in\Total{X_i}\
   \Eset i\times u\subseteq v\}
 \]}%
as easily checked. Notice that the final object is
$\Top=(\emptyset,\Eset\emptyset)$ and that
$\Zero=\Orth\Top=(\emptyset,\emptyset)$.

\subsubsection{Exponential}\label{sec:NUTS-exponential}
We extend the exponential of $\REL$ with totality. If
$u\in\Part{\Web X}$ we set $\Promhc u=\Mfin u\in\Part{\Web{\Excl X}}$.
Then we set $\Web{\Excl X}=\Mfin{\Web X}$ and
\(
  \Total{\Excl X}=\Biorth{\Eset{\Promhc u\St u\in\Total X}}
  =\Upcl{\Eset{\Promhc u\St u\in\Total X}}
\).

\begin{lemma}\label{lemma:nuts-excl-map}
  Let $t\subseteq\Mfin{\Web X}\times\Web Y$. One has
  $t\in\NUTS(\Excl X,Y)$ iff for all $u\in\Total X$ one has
  $\Matappa t{\Promhc u}\in\Total Y$.
\end{lemma}

\begin{lemma}\label{lemma:nuts-excl-map-bil}
  Let
  $t\subseteq\Mfin{\Web{X_1}}\times\cdots\times\Mfin{\Web{X_k}}\times\Web
  Y$. One has $t\in\NUTS(\Excl{X_1}\ITens\cdots\ITens\Excl{X_k},Y)$
  iff for all $\Vect u\in\prod_{i=1}^k\Total{X_i}$ one has
  $\Matappa t{(\Promhc{u_1}\ITens\cdots\ITens\Promhc{u_2})}\in\Total Y$.
\end{lemma}

\begin{lemma}
  If $t\in\NUTS(X,Y)$ then $\Excl t\in\NUTS(\Excl X,\Excl Y)$.
\end{lemma}
\begin{proof}
  By Lemma~\ref{lemma:nuts-excl-map} and the fact that
  $\Matappa{\Excl t}{\Promhc u}=\Promhc{\Matappap tu}$.
\end{proof}

To prove that $\NUTS$ is a categorical model of $\LL$, it suffices to
show that the various relational morphisms defining the strong
symmetric monoidal monadic structure of $\Excl\_$ in $\REL$ (see
Section~\ref{sec:REL-LL-model}) are actually morphisms in $\NUTS$. This
is essentially straightforward and based on
Lemma~\ref{lemma:nuts-excl-map}.

\begin{lemma}\label{lemma:NUTS-excl}
  Equipped with $\Der{}$, $\Digg{}$, $\Seelyz$ and $\Seelyt{}$ defined
  as in $\REL$, $\Excl\_$ is a symmetric monoidal comonad which turns
  $\NUTS$ into a Seely model of $\LL$.
\end{lemma}


\subsection{Variable non-uniform totality spaces (VNUTS)}
\label{sec:VNUTS-definition}
Let $E$ be a set, we use $\Tot E$ for the set of all \emph{totality
  candidates} on $E$, that is, of all subsets $\cT$ of $\Part E$ such
that $\cT=\Biorth\cT$%
.
In other words $\cT\in\Tot E$ means that
$\cT=\Upcl\cT$ by Lemma~\ref{lemma:NUTS-biorth-uppper-closed}. Ordered by
$\subseteq$, this set $\Tot E$ is a complete lattice.

We need now to define a notion of strong functors
$\cX:\NUTS^n\to\NUTS$ for defining a model in the sense of
Definition~\ref{def:categorical-muLL-models}. One crucial feature of
such a functor will be that $\Web{\Strfun\cX(\Vect X)}$ depends only
on $\Web{\Vect X}$.

\begin{definition}
Let $n\in\Nat$, an $n$-ary VNUTS is a pair
$\Vsnot X=(\Web{\Vsnot X},\Total{\Vsnot X})$ where
$\Web{\Vsnot X}:\REL^n\to\REL$ is a variable set
$\Web{\Vsnot X}=(\Strfun{\Web{\Vsnot X}},\Strnat{\Web{\Vsnot X}})$
(see Section~\ref{sec:strong-VS-Seely-model}) and $\Total{\Vsnot X}$
is an operation which with each $n$-tuple $\Vect X$ of objects of
$\NUTS{}$ associates an element $\Total{\Vsnot X}(\Vect X)$ of
$\Tot{\Strfun{\Web{\Vsnot X}}(\Web{\Vect X})}$ in such a way that
\begin{Enumerate}
\item for any $\Vect t\in\NUTS^n(\Vect X,\Vect Y)$, the element
  $\Strfun{\Web{\Vsnot X}}(\Vect t)$ of\\
  $\REL(\Strfun{\Web{\Vsnot X}}({\Web{\Vect X}}),{\Web{\Vsnot
      X}}(\Web{\Vect Y}))$ belongs to
  $\NUTS(\Strfun{\Vsnot X}(\Vect X),\Strfun{\Vsnot X}(\Vect Y))$\\
  (where $\Strfun{\Vsnot X}(\Vect X)$ denotes the NUTS
  $(\Strfun{\Web{\Vsnot X}}(\Web{\Vect X}),\Total{\Vsnot X}(\Vect X))$
  \label{enum:vnuts-cond-tot}
\item and for any $\Vect Y\in\Obj{\NUTS^n}$ and any $X\in\Obj\NUTS$
  one has
  $\Strnat{\Web{\Vsnot X}}_{\Web X,\Web{\Vect Y}}\in\NUTS(\Tens{\Excl
    X}{\Strfun{\Vsnot X}}(\Vect Y),\Strfun{\Vsnot X}(\Tens{\Excl
    X}{\Vect Y}))$. In other words, for an $u\in\Total X$ and
  $v\in\Tot{\Vsnot X}(\Vect Y)$, one has
  $\Matappa{\Strnat{\Web{\Vsnot X}}_{\Web X,\Web{\Vect
        Y}}}{\Tensp{\Promhc u}{w}}\in\Tot{\Vsnot X}(\Tens{\Excl
    X}{\Vect Y})$\label{enum:vnuts-cond-strength}.
\end{Enumerate}
\end{definition}

\begin{lemma}\label{lemma:VNUTS-strong-functor}
  Any VNUTS $\Vsnot X:\NUTS^n\to\NUTS$ induces a strong functor
  $\cX:\NUTS^n\to\NUTS$ which satisfies
  \begin{Itemize}
  \item
    $\Web{\Strfun\cX(\Vect X)}=\Strfun{\Web{\Vsnot
        X}}(\Web{\Vect X})$,
  \item
    $\Total{\Strfun\cX(\Vect X)}=\Total{\Vsnot X}(\Vect X)$,
  \item
    $\Strfun\cX(\Vect t)=\Strfun{\Web{\Vsnot X}}(\Vect
    t)\in\NUTS(\Strfun{\Vsnot X}(\Vect X),\Strfun{\Vsnot X}(\Vect Y))$
    if $\Vect t\in\NUTS^n(\Vect X,\Vect Y)$,
  \item and
    $\Strnat\cX_{X,\Vect Y}=\Strnat{\Web{\Vsnot X}}_{\Web
      X,\Web{\Vect Y}}$
  \end{Itemize}
  and the correspondence $\Vsnot X\mapsto\cX$ is injective.
\end{lemma}
\begin{proof}
  It is clear that $\cX$ so defined is a strong functor.  Let us check
  that $\Vsnot X$ can be retrieved from $\cX$. Given a set $E$,
  $(E,\Part E)$ is a NUTS that we denote as $\Nutsm(E)$. Notice that
  $\Nutsm$ can be extended into a functor $\REL\to\NUTS$ which acts as
  the identity on morphisms. There is also a forgetful functor
  $\Nutsf:\NUTS\to\REL$ which maps $X$ to $\Web X$ and acts as the
  identity on morphisms ($\Nutsm$ is right adjoint to $\Nutsf$). Let
  $\Vsnot X$ be a unary VNUTS and let $\cX:\NUTS\to\NUTS$ be the
  associated strong functor. Then we have
  $\Strfun{\Web{\Vsnot X}}=\Nutsf\Comp\Strfun\cX\Comp\Nutsm$ and
  $\Strnat{\Web{\Vsnot X}}_{E,F}=\Strnat\cX_{\Nutsm(E),\Nutsm(F)}$ for
  any sets $E$ and $F$. Last, given a NUTS $X$, we have that
  $\Total{\Vsnot X}(X)$ is just the totality component of the NUTS
  $\cX(X)$. This shows that the VNUTS which induces $\cX$ can be
  retrieved from $\cX$.
\end{proof}

For this reason we use $\Vsnot X$ to denote the functor $\cX$.

\begin{remark}
  Another possibility would be to define a VNUTS at the first place as
  a strong functor $\cX:\NUTS^n\to\NUTS$ satisfying additional
  properties whose purpose would be to make the definition of the
  underlying $\Vsnot X$ possible. This option, suggested by the
  reviewers, will be explored further. It is crucial to notice that
  these additional properties (that is, the existence of $\Vsnot X$)
  are crucial in the proof of Theorem~\ref{th:VNUTS-model}. In
  particular, it is essential that $\Web{\cX(X)}$ depends only on
  $\Web X$.
\end{remark}

Given $n\in\Nat$ let $\VNUTS n$ be the class of strong $n$-ary VNUTS.
We identify $\VNUTS 0$ with the class of objects of the Seely category
$\NUTS$. The following refers to
Definition~\ref{def:categorical-muLL-models}
\begin{theorem}\label{th:VNUTS-model}
  $(\NUTS,(\VNUTS n)_{n\in\Nat})$ is a Seely model of $\MULL$.
\end{theorem}
\begin{proof}[Partial proof]
We deal with Condition~(\ref{enum:seel-mull-5}).\\
%
Let first $\Vsnot X=(\Web{\Vsnot X},\Total{\Vsnot X})$
be a unary VNUTS. Let $E=\Funfp{\Strfun{\Web{\Vsnot X}}}$ which is the
least set such that $\Strfun{\Web{\Vsnot X}}(E)=E$, that is
$E=\Union_{n=0}^\infty\Strfun{\Web{\Vsnot X}}^n(\emptyset)$. Let
$\Phi:\Tot E\to\Tot E$ be defined as follows: given $\cS\in\Tot E$,
then $(E,\cS)$ is a NUTS, and we set
$\Phi(\cS)=\Total{\Vsnot X}(E,\cS)\in\Tot{\Strfun{\Web{\Vsnot
      X}}(E)}=\Tot E$. This function $\Phi$ is monotone. Let indeed
$\cS_1,\cS_2\in\Tot E$ with $\cS_1\subseteq\cS_2$. Then we have
$\Id\in\NUTS((E,\cS_1),(E,\cS_2))$ and therefore, by
Condition~(\ref{enum:vnuts-cond-tot}) satisfied by $\Vsnot X$, we have
$\Id=\Strfun{\Web{\Vsnot X}}(\Id)\in\NUTS(\Strfun{\Vsnot
  X}(E,\cS_1),\Strfun{\Vsnot
  X}(E,\cS_2))=\NUTS((E,\Phi(\cS_1)),(E,\Phi(\cS_2))$ which means that
$\Phi(\cS_1)\subseteq\Phi(\cS_2)$. By the Knaster Tarski Theorem (remember
that $\Tot E$ is a complete lattice), $\Phi$ has a greatest fixpoint
$\cT$ that we can describe as follows. Let
$(\cT_\alpha)_{\alpha\in\Ordinals}$, where $\Ordinals$ is the class of
ordinals, be defined by: $\cT_0=\Part E$ (the largest possible notion
of totality on $E$), $\cT_{\alpha+1}=\Phi(\cT_\alpha)$ and
$\cT_\lambda=\Inter_{\alpha<\lambda}\cT_\alpha$ when $\lambda$ is a
limit ordinal. This sequence is decreasing (easy induction on ordinals
using the monotonicity of $\Phi$) and there is an ordinal $\theta$
such that $\cT_{\theta+1}=\cT_\theta$ (by a cardinality argument; we
can assume that $\theta$ is the least such ordinal). The greatest
fixpoint of $\Phi$ is then $\cT_\theta$ as easily checked.

By construction $((E,\cT_\theta),\Id)$ is an object of
$\COALGFUN{\NUTS}{\Strfun{\Vsnot X}}$, we prove that it is the
final object. So let $(Y,t)$ be another object of the same
category. Since $(\Web Y,t)$ is an object of
$\COALGFUN\REL{\Strfun{\Web{\Vsnot X}}}$ and since $(E,\Id)$ is the
final object in that category, we know by
Lemma~\ref{lemma:rel-fixpoint-final} that there is exactly one
$e\in\REL(\Web Y,E)$ such that $\Strfun{\Web{\Vsnot X}}(e)\Compl
t=e$. We prove that actually $e\in\NUTS(Y,(E,\cT_\theta))$ so let
$v\in\Total Y$. We prove by induction on the ordinal $\alpha$ that
$\Matappa ev\in\cT_\alpha$. For $\alpha=0$ it is obvious since
$\cT_0=\Part E$. Assume that the property holds for $\alpha$ and let
us prove it for $\alpha+1$. We have
$\Matappa tv\in\Total{\Vsnot X}(Y)=\Total{\Strfun{\Vsnot X}(Y)}$ since
$t\in\NUTS(Y,\Strfun{\Vsnot X}(Y))$. Since
$\Strfun{\Vsnot X}(e)\in\NUTS(\Strfun{\Vsnot X}(Y),\Strfun{\Vsnot
  X}(E,\cT_\alpha))$ and
$\Strfun{\Vsnot X}(E,\cT_\alpha)=(E,\cT_{\alpha+1})$ we have
$\Matappa{(\Strfun{\Vsnot X}(e)\Compl t)}v\in\cT_{\alpha+1}$, that is
$\Matappa ev\in\cT_{\alpha+1}$. Last if $\lambda$ is a limit ordinal
and if we assume $\forall\alpha<\lambda\ \Matappa ev\in\cT_\alpha$ we
have $\Matappa
ev\in\Inter_{\alpha<\lambda}\cT_\alpha=\cT_\lambda$. Therefore
$\Matappa ev\in\cT_\theta$. We use $\Fungfp{\Strfun{\Vsnot X}}$ to
denote this final coalgebra $(E,\cT_\theta)$ (its definition depends
only on $\Strfun{\Vsnot X}$ and does not involve the strength
$\Strnat{\Vsnot X}$).

So we have proven the first part of Condition~(\ref{enum:seel-mull-5})
in the definition of a Seely model of $\MULL$ (see
Section~\ref{def:categorical-muLL-models}). As to the second part, let
$\Vsnot X$ be an $n+1$-ary VNUTS. We know by the general
Lemma~\ref{lemma:strfun-gfp-general} how to build a 
strong functor $\Fungfp{\Vsnot X}:\NUTS^n\to\NUTS$ with suitable properties.
To end the proof, it suffices to exhibit an $n$-ary VNUTS
$\Vsnot Y=(\Web{\Vsnot Y},\Total{\Vsnot Y})$ whose associated strong
functor coincides with $\Fungfp{\Vsnot X}$. The construction of
$\Vsnot Y$ is essentially straightforward, using the constructions
available in $\REL$.
\end{proof}

\begin{remark}
  For any closed formula $A$, the web of its interpretation $\Tsemn A$
  in $\NUTS$ coincides with its interpretation $\Tsemr A$ in
  $\REL$. It is also easy to check that for any proof $\pi$ of
  $\Seq A$, one has $\Psemn\pi=\Psemr\pi$ (this can be formalized
  using the functor $\Nutsf:\NUTS\to\REL$ introduced in the proof of
  Lemma~\ref{lemma:VNUTS-strong-functor}, which acts trivially on
  morphisms).
\end{remark}
\todo[inline]{I've removed the remark about topological functors which
  I found useless.}
\begin{remark}
  The same method can be applied in many contexts. For instance, we
  can replace $\REL$ with the category of coherence spaces --~where
  least and greatest fixpoints are interpreted in the same way~-- and
  $\NUTS$ with the category of coherence spaces with totality where
  the interpretations will be different. One of the reviewers
  suggested that this situation might be generalized using the concept
  of \emph{topological functors}, this option will be explored in
  further work.
\end{remark}


\subsection{Examples of data-types}

\subsubsection{Integers}\label{sec:example-integers}
The type of ``flat integers'' is defined by
$\Tnat=\Lfpll\zeta{(\Plus\One\zeta)}$. In $\REL$, $\Plus 1\zeta$ is
interpreted as the unary variable set
$\Tsemr{\Plus\One\zeta}_\zeta:\REL\to\REL$ which maps a set $E$ to
$\Plus\One E=\Eset{(1,\Onelem)}\cup(\Eset 2\times E)$. Hence
$\Tsemr\Tnat$ is the least set such that
$\Tsem\Tnat=\Eset{(1,\Onelem)}\cup(\Eset 2\times \Tsem\Tnat)$ that is,
the set of all tuples $\Snum n=(2,(2,(\cdots(1,\Onelem)\cdots)))$
where $n$ is the number of occurrence of $2$ so that we can consider
the elements of $\Tsem\Tnat$ as integers. We have
$\Web{\Tsemn\Tnat}=\Tsemr\Tnat$ and we compute $\Total{\Tsemn\Tnat}$
dually wrt.~the proof of Theorem~\ref{th:VNUTS-model}: it is the least
fixed point of the operator
$\Phi:\Tot{\Tsemr\Tnat}\to\Tot{\Tsemr\Tnat}$ such that, if
$\cT\in\Tot{\Tsemr\Tnat}$ then
$\Phi(\cT)=\Eset{u\subseteq\Tsemr\Tnat\St\Snum 0\in u\text{ or
  }\Eset{\Snum n\in\Tsemr\Tnat\St\Snum{n+1}\in u}\in\cT}$. Therefore
$\Tot{\Tsemn\Tnat}=\Eset{u\subseteq\Tsemr\Tnat\St u\not=\emptyset}$.
\begin{theorem}
  If $\pi$ is a proof of $\Seq\Tnat$ then $\Psemn\pi=\Psemr\pi$ is a
  non-empty subset of $\Tsemr\Tnat$.
\end{theorem}
Indeed we know that $\Psemr\pi=\Psemn\pi\in\Total{\Tsemn\Tnat}$. Using
an additional notion of coherence (which can be fully compatible with
$\REL$ as in the non-uniform coherence space models
of~\cite{BucciarelliEhrhard99,Boudes11}) we can also prove that
$\Psemr\pi$ has at most one element, and hence is a singleton
$\Eset n$. This is a denotational version of normalization expressing
that indeed $\pi$ ``has a value'' (and actually exactly one, which
expresses a weak form of confluence).

\subsubsection{Binary trees with integer leaves} This type can be
defined as $\tau=\Lfpll\zeta{(\Tnat\IPlus\Tensp\zeta\zeta)}$. Then an
element of $\Tsemr\tau=\Web{\Tsemn\tau}$ is an element of the set
described by the following syntax:
$\alpha,\beta,\cdots\Bnfeq\Leaf n\Bnfor\Bnode\alpha\beta$. A
computation similar to the previous one shows that
$\Tot{\Tsemn\tau}=\Eset{u\subseteq\Tsemr\tau\St u\not=\emptyset}$.

\subsubsection{An empty type of streams of integers} After
reading~\cite{BaeldeDoumaneSaurin16}, one could be tempted to define
the type of streams of integers as
$\sigma_0=\Gfpll\zeta{(\Tens\Tnat\zeta)}$. The variable set
$\Tsemr{\Tens\Tnat\zeta}_\zeta:\REL\to\REL$ maps a set $E$ to
$\Nat\times E$. The least fixed point of this operation on sets is
$\emptyset$ and hence $\Web{\Tsemn\sigma}=\emptyset$ and notice that
$\Tot\emptyset=\Eset{\emptyset,\Eset\emptyset}$. In that case, the
operation $\Phi:\Tot{\emptyset}\to\Tot\emptyset$ maps $\cT$ to
$\Eset{u\times v\St v\in\cT\text{ and
  }u\in\Part\Nat\setminus\Eset\emptyset}$ and hence $\Eset\emptyset$
to itself. It follows that $\Total{\Tsemn{\sigma_0}}=\Eset\emptyset$
that is $\Tsemn{\sigma_0}=\Top$, the final object of $\NUTS$. What is the
meaning of this trivial interpretation? It simply reflects that,
though $\sigma_0$ has a lot of non trivial proofs in $\MULL$, it is
impossible to extract any finite information from these proofs within
$\MULL$, and accordingly all these proofs are interpreted as
$\emptyset$. 
\begin{theorem}
  In $\MULL$, there is no proof of $\Seq{\Orth{\sigma_0},\Tnat}$.
\end{theorem}
In other words there is no proof of $\Seq{\Limpl{\sigma_0}\Tnat}$ in
$\MULL$; typically a function extracting the first element of a stream
would be a proof of this type\dots{} if it would exist! Here is the
argument: if $\pi$ were a proof of $\Seq{\Orth{\sigma_0},\Tnat}$, we
would have $\Psem\pi\in\NUTS(\Tsemn{\sigma_0},\Tsemn\Tnat)$ and hence
$\Matappa{\Psem\pi}\emptyset\in\Total{\Tsemn\Tnat}$ which is not the
case since $\Matappa{\Psem\pi}\emptyset=\emptyset$. If types
like $\sigma_0$ are meaningful in a proof-search perspective, their
relevance as data-types in a Curry-Howard approach to $\MULL$ is
dubious.

\subsubsection{A non-empty type of streams of integers} We set now
$\sigma=\Gfpll\zeta{(\One\Bwith\Tensp{\Tnat}{\zeta})}$. This type
looks like the previous one, but the type $\One$ leaves space
for \emph{partial} empty streams. Warning: it is \emph{not} a type of
finite or infinite streams; the $\IWith$ means that this empty stream
will not be a total element: it will have to be complemented by some
total element from the right argument of the $\IWith$. More precisely
$\Tsemr{\One\IWith\Tensp\Nat\zeta}_\zeta:\REL\to\REL$ is the variable
set which maps a set $E$ to
$\Eset{(1,\Onelem)}\cup\Eset 2\times\Nat\times E$ so that up to
renaming $\Web{\Tsemn\sigma}=\Fseq\Nat$ (all \emph{finite sequences}
of integers). In this case, the operator
$\Phi:\Tot{\Fseq\Nat}\to\Tot{\Fseq\Nat}$ maps $\cT$ to
\[
  \Eset{v\subseteq\Fseq\Nat\St\Emptyseq\in v\text{ and }\exists
    n\in\Nat,u\in\cT\ \Eset n\times u\subseteq v}
\]
where we use
$\Emptyseq$ for the empty sequence. So for instance
{\footnotesize%
  \begin{align*}
  \Phi^0(\Part{\Fseq\Nat})&=\Part{\Fseq\Nat}\quad
  \Phi^1(\Part{\Fseq\Nat})=\Eset{u\in\Part{\Fseq\Nat}\St\Emptyseq\in
                            u}\\
  \Phi^3(\Part{\Fseq\Nat})&=\Eset{u\in\Part{\Fseq\Nat}\St \exists
                            n_1,n_2\ \Emptyseq,(n_1),(n_1,n_2)\in u}\,.
  \end{align*}%
}
The greatest fixed
point is reached in $\omega$ steps:
\begin{multline*}
  \Tot{\Tsemn{\sigma}}
  =\Inter_{n<\omega}\Phi^n(\Part{\Fseq\Nat})\\
  =\Eset{u\subseteq\Fseq\Nat\St\exists f\in\Nat^\omega\,\forall
  k<\omega\ (f(1),\dots,f(k))\in u}\,.
\end{multline*}
So a total subset of $\Web{\Tsemn\sigma}$ must contain (at least) an
infinite stream of integer. For this type of streams $\sigma$ it is
easy to build a proof of $\Seq{\Orth\sigma,\Tnat}$ extracting the
first element of a stream, interpreted as
$\Eset{((n),n)\St n\in\Nat}$.

\section{Conclusion and further work}\label{sec:conclusion}

We will study next the semantics of infinite proofs of $\MULL$ (whose
definition extends that of~\cite{BaeldeDoumaneSaurin16} for
$\MUMALL$). A crucial step is to prove that these infinite proofs can
be interpreted as total sets in $\NUTS$, this will be presented in a
further paper. This interpretation of proofs is based on the
interpretation of their finite approximations in $\REL$ (remember that
the interpretations of a $\MULL$ proof in $\NUTS$ and in $\REL$ are
\emph{exactly the same set}).

Our models will also serve as guidelines for the design of a
functional language based on $\MULL$, generalizing Gödel's System T in
the spirit of~\cite{Loader97} though, as explained in the
Introduction, Loader's syntax is not fully compatible with $\LL$ as it
is based on \emph{cocartesian} cartesian closed categories. Our system
will primarily implement Park's rule, but we will also consider other options
based on polymorphism in the spirit of~\cite{Mendler91,CamposFiore20}
or~\cite{Matthes98}, or general recursion with guardedness
restrictions as in~\cite{Coquand93,Paulin-Mohring93,Gimenez98}.

Its syntax will be based on the idea of representing data-types as
\emph{positive} formulas of $\MULL$ interpreted in $\Em\cL$ and
therefore equipped with morphisms implementing weakening, contraction
and promotion: as noticed in~\cite{Baelde12}, $\Lfpll\zeta\_$ is a
positive operation whereas $\Gfpll\zeta\_$ is negative. In $\Em\cL$,
the $\IPlus$ of $\LL$ is a coproduct and the $\ITens$ is a cartesian
product as expected (and $\ITens$ distributes over $\IPlus$). The
targeted calculus will feature a notion of values (positive terms)
accounting for the morphisms of $\Em\cL$, substitution in terms being
allowed only for values because only them can safely be discarded and
duplicated. Thanks to this choice of design, weakening and contraction
will remain implicit operations as in the usual
$\lambda$-calculus. Our calculus will have explicit promotion and
dereliction operations, allowing to implement both CBN and CBV in the
same setting, just as in Levy's
Call-by-push-value~\cite{LevyP06,Ehrhard16a}.

We thank the reviewers of this paper for their careful reading and
very useful suggestions. This work was partly funded by the ANR
project PPS, ANR-19-CE48-0014. This paper is a preprint version of an
article published at LICS'21.

%


\bibliography{newbiblio}

\section{Appendix}

\subsection{Proof of Lemma~\ref{lemma:functor-gfp-general}}
\begin{proof} We have
$\cF(g,\Fungfp\cF(B))\in\cA(\Fungfp\cF(B),\cF(B',\Fungfp\cF(B)))$ thus
defining an $\cF_{B'}$-coalgebra structure on $\Fungfp\cF(B)$ and hence
there exists a unique morphism $\Fungfp\cF(g)$ such that
\[
  \cF(B',\Fungfp\cF(g))\Compl\cF(g,\Fungfp\cF(B))=\Fungfp\cF(g)\,,
\]
that is $\cF(g,\Fungfp\cF(g))=\Fungfp\cF(g)$.

Functoriality follows: consider also $g'\in\cB(B',B'')$, then we know
that $h=\Fungfp\cF(g'\Compl g)$ satisfies $\cF(g'\Compl g,h)=h$ by the
definition above. Now $h'=\Fungfp\cF(g')\Compl\Fungfp\cF(g)$ satisfies
the same equation by functoriality of $\cF$ and because
$\cF(g,\Fungfp\cF(g))=\Fungfp\cF(g)$ and
$\cF(g',\Fungfp\cF(g'))=\Fungfp\cF(g')$, and hence $h'=h$ by
Lemma~\ref{lemma:equations-final-coalgebra}, taking
$l=\cF(g'\Compl g,\Fungfp\cF(B))$. In the same way one proves that
$\Fungfp\cF(\Id)=\Id$. \end{proof}

\subsection{Proof of Lemma~\ref{lemma:strfun-gfp-general}}
\begin{proof} The part of the statement which concerns the functor
$\Strfun{\Fungfp{\Vsnot F}}$ is a direct application of
Lemma~\ref{lemma:functor-gfp-general} so we only have to deal with the
strength.  Let us prove naturality so let
$\Vect f\in\LCAT^n(\Vect X,\Vect{X'})$ and $g\in\LCAT(Y,Y')$, we must
prove that the following diagram commutes
\[
  \begin{tikzpicture}[->, >=stealth]
      \node (1) {$\Tens{\Excl Y}{\Strfun{\Fungfp{\Vcsnot F}}(\Vect X)}$};
      \node (2) [right of=1, node distance=5.4cm]
      {$\Strfun{\Fungfp{\Vcsnot F}}(\Tens{\Excl Y}{\Vect X})$};
      \node (3) [below of=1, node distance=1.2cm] 
      {$\Tens{\Excl {Y'}}{\Strfun{\Fungfp{\Vcsnot F}}(\Vect{X'})}$};
      \node (4) [below of=2, node distance=1.2cm] 
      {$\Strfun{\Fungfp{\Vcsnot F}}(\Tens{\Excl {Y'}}{\Vect{X'}})$};
       \tikzstyle{every node}=[midway,auto,font=\scriptsize]
       \draw (1) -- node {$\Strnat{\Fungfp{\Vcsnot F}}_{Y,\Vect X}$} (2);
       \draw (1) -- node [swap]  
       {$\Tens{\Excl g}{\Strfun{\Fungfp{\Vcsnot F}}(\Vect f)}$} (3);
       \draw (2) -- node
       {$\Strfun{\Fungfp{\Vcsnot F}}(\Tens{\Excl g}{\Vect f})$} (4);
       \draw (3) -- node {$\Strnat{\Fungfp{\Vcsnot F}}_{Y',\Vect{X'}}$} (4);
  \end{tikzpicture}   
\]
Let
$h_1=\Strnat{\Fungfp{\Vcsnot F}}_{Y',\Vect{X'}} \Compl(\Tens{\Excl
  g}{\Strfun{\Fungfp{\Vcsnot F}}(\Vect f)})$ and
$h_2=\Strfun{\Fungfp{\Vcsnot F}}(\Tens{\Excl g}{\Vect
  f})\Compl\Strnat{\Fungfp{\Vcsnot F}}_{Y,\Vect X}$ be the two
morphisms we must prove equal. We
use Lemma~\ref{lemma:equations-final-coalgebra}, taking the following
morphism $l$.
\begin{equation*}
    \begin{tikzpicture}[->, >=stealth]
      \node (1) {$\Tens{\Excl Y}{\Strfun{\Fungfp{\Vcsnot F}}(\Vect X)}
        =\Excl Y\ITens
          \Strfun{\Vcsnot F}(\Vect X,\Strfun{\Fungfp{\Vcsnot F}}(\Vect X))$};
      \node (2) [below of=1, node distance=1.2cm]
      {$\Strfun{\Vcsnot F}(\Tens{\Excl Y}{\Vect X},
        \Tens{\Excl Y}{\Strfun{\Fungfp{\Vcsnot F}}(\Vect X)})$};
      \node (3) [below of=2, node distance=1.2cm]
      {$\Strfun{\Vcsnot F}(\Tens{\Excl {Y'}}{\Vect{X'}},
        \Tens{\Excl Y}{\Strfun{\Fungfp{\Vcsnot F}}(\Vect X)})$};
       \tikzstyle{every node}=[midway,auto,font=\scriptsize]
       \draw (1) -- node {$\Strnat{\Vcsnot F}_{Y,(\Vect X,
           \Strfun{\Fungfp{\Vcsnot F}}(\Vect X))}$} (2);
       \draw (2) -- node {$\Strfun{\Vcsnot F}(\Tens{\Excl g}
         {\Vect f},\Id)$} (3);
     \end{tikzpicture}   
\end{equation*}
With these notations we have
\begin{align*}
  &\Strfun{\Vcsnot F}(\Tens{\Excl {Y'}}{\Vect{X'}},h_1)\Compl l\\
  &= \Strfun{\Vcsnot F}
    (\Tens{\Excl {Y'}}{\Vect{X'}},\Strnat{\Fungfp{\Vcsnot F}}_{Y',\Vect{X'}})
    \Compl \Strfun{\Vcsnot F}(\Tens{\Excl {Y'}}{\Vect{X'}},
    \Tens{\Excl g}{\Strfun{\Fungfp{\Vcsnot F}}(\Vect f)})\\
  &\quad\quad \Compl \Strfun{\Vcsnot F}(\Tens{\Excl g}{\Vect f},
    \Tens{\Excl Y}{\Strfun{\Fungfp{\Vcsnot F}}(\Vect X)})
    \Compl \Strnat{\Vcsnot F}_{Y,(\Vect X,\Strfun{\Fungfp{\Vcsnot F}}(\Vect X))}\\
  &= \Strfun{\Vcsnot F}
    (\Tens{\Excl {Y'}}{\Vect{X'}},\Strnat{\Fungfp{\Vcsnot F}}_{Y',\Vect{X'}})
    \Compl \Strfun{\Vcsnot F}(\Tens{\Excl g}{\Vect f},
    \Tens{\Excl g}{\Strfun{\Fungfp{\Vcsnot F}}(\Vect f)})
    \Compl \Strnat{\Vcsnot F}_{Y,(\Vect X,\Strfun{\Fungfp{\Vcsnot F}}(\Vect X))}\\
  &= \Strfun{\Vcsnot F}
    (\Tens{\Excl {Y'}}{\Vect{X'}},\Strnat{\Fungfp{\Vcsnot F}}_{Y',\Vect{X'}})
    \Compl \Strnat{\Vcsnot F}_{Y',(\Vect{X'},\Strfun{\Fungfp{\Vcsnot F}}(\Vect{X'}))}
    \Compl (\Tens{\Excl g}{\Strfun{\Vcsnot F}}
    (\Vect f,\Strfun{\Fungfp{\Vcsnot F}}(\Vect f)))\\
  &\quad\quad\quad\quad\quad\quad \text{ by naturality of }\Strnat{\Vcsnot F}\\
  &= \Strnat{\Fungfp{\Vcsnot F}}_{Y',\Vect{X'}}
    \Compl(\Tens{\Excl g}{\Strfun{\Vcsnot F}}
    (\Vect f,\Strfun{\Fungfp{\Vcsnot F}}(\Vect f)))
    \text{ by~\Eqref{eq:final-coalg-strength-charact}}\\
  &= \Strnat{\Fungfp{\Vcsnot F}}_{Y',\Vect{X'}}
    \Compl(\Tens{\Excl g}{\Strfun{\Fungfp{\Vcsnot F}}(\Vect f)})
    \text{ by~Lemma~\ref{lemma:functor-gfp-general}}
\end{align*}
so that $\Strfun{\Vcsnot F}(\Tens{\Excl {Y'}}{\Vect{X'}},h_1)\Compl l=h_1$
as required. On the other hand we have
\begin{align*}
  &\Strfun{\Vcsnot F}(\Tens{\Excl {Y'}}{\Vect{X'}},h_2)\Compl l\\
  &= \Strfun{\Vcsnot F}(\Tens{\Excl {Y'}}{\Vect{X'}},
    \Strfun{\Fungfp{\Vcsnot F}}(\Tens{\Excl g}{\Vect f}))
  \Compl \Strfun{\Vcsnot F}(\Tens{\Excl {Y'}}{\Vect{X'}},
    \Strnat{\Fungfp{\Vcsnot F}}_{Y,\Vect X})\\
  &\quad\quad \Compl \Strfun{\Vcsnot F}(\Tens{\Excl g}{\Vect f},
    \Tens{\Excl Y}{\Strfun{\Fungfp{\Vcsnot F}}(\Vect X)})
    \Compl \Strnat{\Vcsnot F}_{Y,(\Vect X,\Strfun{\Fungfp{\Vcsnot F}}(\Vect X))}\\
  &= \Strfun{\Vcsnot F}(\Tens{\Excl {Y'}}{\Vect{X'}},
    \Strfun{\Fungfp{\Vcsnot F}}(\Tens{\Excl g}{\Vect f}))
    \Compl \Strfun{\Vcsnot F}(\Tens{\Excl g}{\Vect f},
    \Tens{\Excl Y}{\Strfun{\Fungfp{\Vcsnot F}}(\Vect X)})\\
  &\quad\quad \Compl \Strfun{\Vcsnot F}(\Tens{\Excl Y}{\Vect X},
    \Strnat{\Fungfp{\Vcsnot F}}_{Y,\Vect X})
    \Compl \Strnat{\Vcsnot F}_{Y,(\Vect X,\Strfun{\Fungfp{\Vcsnot F}}(\Vect X))}\\
  &= \Strfun{\Vcsnot F}(\Tens{\Excl g}{\Vect f},
    \Strfun{\Fungfp{\Vcsnot F}}(\Tens{\Excl g}{\Vect f}))
    \Compl\Strnat{\Fungfp{\Vcsnot F}}_{Y,\Vect X}
    \text{ by~\Eqref{eq:final-coalg-strength-charact}}\\
  &=  \Strfun{\Fungfp{\Vcsnot F}}(\Tens{\Excl g}{\Vect f})
    \Compl\Strnat{\Fungfp{\Vcsnot F}}_{Y,\Vect X}
    \text{ by~Lemma~\ref{lemma:functor-gfp-general}}
\end{align*}
so that
$\Strfun{\Vcsnot F}(\Tens{\Excl {Y'}}{\Vect{X'}},h_2)\Compl l=h_2$
which proves our contention. The commutation of the diagrams of
Figure~\ref{fig:strength-monoidality} for
$\Strnat{\Fungfp{\Vcsnot F}}$ is proven similarly.%
\end{proof}

\subsection{Proof of Lemma~\ref{lemma:rel-embedding-retraction}}
\begin{proof}
  Let $a\in E$, since $(a,a)\in\Id_E=t\Compl s$, there must exist
  $b\in F$ such that $(a,b)\in s$ and $(b,a)\in t$. If $(a,b')\in s$
  then $(b,b')\in s\Compl t\subseteq\Id_F$ and hence $b'=b$. It
  follows that $s$ is a total function $E\to F$. Let $(a,b)\in s$
  (that is $a\in E$ and $b=s(a)$). Since $t\Compl s=\Id_E$, we must
  have $(b,a)\in t$. Conversely let $(b,a)\in t$, we have
  $(b,s(a))\in s\Compl t$ and hence $b=s(a)$. We have proven that
  $t=\Eset{(s(a),a)\St a\in E}$. If $a,a'\in $ satisfy $s(a)=s(a')$ we
  have therefore $(a,a')\in t\Compl s=\Id_E$ and hence $a=a'$; this
  shows that $s$ is injective.
\end{proof}

\subsection{Proof of Lemma~\ref{lemma:rel-fixpoint-final}}
\begin{proof}
  Let $(E,t)$ be an $\Vsnot F$-coalgebra.
  We define a sequence $e_n\in\REL(E,\Funfp{\Vsnot F})$ as follows:
  $e_0=\emptyset$ and $e_{n+1}=\Vsnot F(e_n)\Compl t$. Then
  $e_n\subseteq e_{n+1}$ for all $n$ by an easy induction, using the
  fact that $\Vsnot F$ is locally continuous. Let
  $e=\Union_{n=0}^\infty e_n\in\REL(E,\Funfp{\Vsnot F})$, by
  local continuity of $\Vsnot F$ we have
  \(
    \Vsnot F(e)\Compl t
    = \left(\Union_{n=0}^\infty\Vsnot F(e_n)\right)\Compl t
    = \Union_{n=0}^\infty (\Vsnot F(e_n)\Compl t)
    = \Union_{n=0}^\infty e_{n+1}=e
    \)
  which means that
  \[
    e\in\COALGFUN\REL{\Vsnot F}((E,t),(\Funfp{\Vsnot F},\Id))\,.
  \]
  We end the proof by showing that $e$ is the unique such morphism, so
  let
  \(e'\in\COALGFUN\REL{\Vsnot F}((E,t),(\Funfp{\Vsnot F},\Id))\) which
  means that $e'\in\REL(E,\Funfp{\Vsnot F})$ and
  $\Vsnot F(e')\Compl t=e'$.

  Let $i_n\in\REL(\Funfp{\Vsnot F},\Funfp{\Vsnot F})$ be defined by
  induction by $i_0=\emptyset$ and $i_{n+1}=\Vsnot F(i_n)$. Then
  $(i_n)_{n\in\Nat}$ is monotone and $\Union_{n=0}^\infty i_n=\Id$ by
  definition of $\Funfp{\Vsnot F}$. We prove by induction on $n$ that
  $\forall n\in\Nat\ i_n\Compl e'=i_n\Compl e$.  Clearly
  $i_0\Compl e'=i_0\Compl e=\emptyset$. Next
  \begin{align*}
    i_{n+1}\Compl e'
    &= \Vsnot F(i_n)\Compl \Vsnot F(e')\Compl t\\
    &= \Vsnot F(i_n\Compl e')\Compl t\\
    &= \Vsnot F(i_n\Compl e)\Compl t\text{\quad by inductive hypothesis}\\
    &= i_{n+1}\Compl e\,.
  \end{align*}
  Therefore
  $e'=\left(\Union_{n\in\Nat}i_n\right)\Compl
  e'=\Union_{n\in\Nat}(i_n\Compl e')=\Union_{n\in\Nat}(i_n\Compl
  e)=e$.
\end{proof}

\subsection{Proof of Proposition~\ref{prop:hom-continuous-dir-cocontinuous}}
\begin{proof}
  Let $\cD$ be a directed set of sets and let $H$ be a set. For each
  $E\in\cD$ let $s_E\in\REL(\Vsnot F(E),H)$ so that $(s_E)_{E\in\cD}$
  defines a cocone, that is, for each $E,F\in\cD$ such that
  $E\subseteq F$, one has $s_E=s_F\Compl\Vsnot F(\Relii_{E,F})$. Let
  $L=\Union\cD$. Let $s\in\REL(\Vsnot F(L),H)$ be given by
  \(
  s=\Union_{E\in\cD}s_E\Compl\Vsnot F(\Relip_{E,L})
  \).
  Let $E\in\cD$, we have
  \(
    s\Compl\Vsnot F(\Relii_{E,L})
    =\Union_{F\in\cD}s_F\Compl\Vsnot F(\Relip_{F,L}\Compl\Relii_{E,L})
    \)
  so that $s_E\subseteq s\Compl\Vsnot F(\Relii_{E,L})$ (since
  $s_F\Compl\Vsnot F(\Relip_{F,L}\Compl\Relii_{E,L})=s_E$ when $F=E$).

  We prove the converse inclusion. Let $F\in\cD$ and let $G\in\cD$ be
  such that $E,F\subseteq G$ (remember that $\cD$ is directed). We
  have
  \begin{align*}
    &s_F\Compl\Vsnot F(\Relip_{F,L}\Compl\Relii_{E,L})
    =s_F\Compl\Vsnot F(\Relip_{F,G}\Compl\Relip_{G,L}
      \Compl\Relii_{G,L}\Compl\Relii_{E,G})\\
    &\quad=s_F\Vsnot F(\Relip_{F,G}\Compl\Relii_{E,G})
    =s_G\Vsnot F(\Relii_{F,G})\Compl\Vsnot F(\Relip_{F,G}\Compl\Relii_{E,G})\\
    &\quad\subseteq s_G\Compl\Vsnot F(\Relii_{E,G})=s_E
  \end{align*}
  where we have used the fact that
  $\Relii_{F,G}\Compl\Relip_{F,G}\subseteq\Id_G$ and hence
  $\Vsnot F(\Relii_{F,G}\Compl\Relip_{F,G})\subseteq\Id_{\Vsnot F(G)}$
  by local continuity of $\Vsnot F$.

  So $s_F\Compl\Vsnot F(\Relip_{F,L}\Compl\Relii_{E,L})\subseteq s_E$
  for all $F\in\cD$ and hence
  $s\Compl\Vsnot F(\Relii_{E,L})\subseteq s_E$ as contended.

  Let now $s'\in\REL(\Vsnot F(L),H)$ be such that
  $s'\Compl\Vsnot F(\Relii_{E,L})=s_E$ for each $E\in\cD$, we show
  that $s'=s$ thus proving the uniqueness part of the universal
  property. For $E\in\cD$, let
  $\theta_E=\Relii_{E,L}\Relip_{E,L}\in\REL(L,L)$. Then
  $(\theta_E)_{E\in\cD}$ is a directed family (for $\subseteq$) and
  $\Union_{E\in\cD}\theta_E=\Id_L$. By local continuity of $\Vsnot F$,
  we have
  \begin{align*}
    s'
    &=s'\Compl\Id_{\Vsnot F(L)}
    =s'\Compl\Union_{E\in\cD}\Vsnot F(\theta_E)\\
    &=\Union_{E\in\cD}s'\Compl\Vsnot F(\Relii_{E,L})
      \Compl\Vsnot F(\Relip_{E,L})
    =\Union_{E\in\cD}s_E\Compl\Vsnot F(\Relip_{E,L})=s
  \end{align*}
  by our assumption on $s'$ and by definition of $s$. This shows that
  the cocone $(\Vsnot F(\Relii_{E,L}))_{E\in\cD}$ on
  $\Vsnot F\Compl\Relii$ is colimiting, thus proving that
  $\Vsnot F\Compl\Relii$ is directed cocontinuous.
\end{proof}

\subsection{Proof of Lemma~\ref{lemma:hom-conts-stable-fixpoint}}
\begin{proof}
  As usual we assume that $n=1$ to increase readability. 
%
  We need to prove first that $\Fungfp{\Vsnot F}$ is monotone on
  morphisms, so let $s,t\in\REL(E,F)$ with $s\subseteq t$. We have
  $\Fungfp{\Vsnot F}(s)=\Union_{n\in\Nat}s_n$ and
  $\Fungfp{\Vsnot F}(t)=\Union_{n\in\Nat}t_n$ with
  $s_0=t_0=\emptyset$, $s_{n+1}=\Vsnot F(s,s_n)$ and
  $t_{n+1}=\Vsnot F(t,t_n)$ (we use the action of $\Fungfp{\Vsnot F}$
  on morphisms resulting from Lemma~\ref{lemma:functor-gfp-general}
  and from the characterization of the morphisms to the final object
  given in the proof of Lemma~\ref{lemma:rel-fixpoint-final}). By
  induction and hom-monotonicity of $\Vsnot F$ we have
  $\forall n\in\Nat\ s_n\subseteq t_n$ and hence
  $\Fungfp{\Vsnot F}(s)\subseteq\Fungfp{\Vsnot F}(t)$.  Let us prove
  now local continuity so let $D\subseteq\REL(E,F)$ be directed and let
  $t=\Union D$, we prove that
  $\Fungfp{\Vsnot F}(t)=\Union_{s\in D}\Fungfp{\Vsnot
    F}(s)\in\REL(\Fungfp{\Vsnot F}(E),\Fungfp{\Vsnot F}(F))$ using
  Lemma~\ref{lemma:equations-final-coalgebra} (with the notations of
  that lemma, we take $l=\Vsnot F(t,\Fungfp{\Vsnot F}(E))$).
  We have
  \[
    \Vsnot F_F(\Fungfp{\Vsnot F}(t))\Compl\Vsnot F(t,\Fungfp{\Vsnot
      F}(E))=\Fungfp{\Vsnot F}(t)
  \]
  by definition of the functor $\Fungfp{\Vsnot F}$ and
  \begin{align*}
    &\Vsnot F_F(\Union_{s\in D}\Fungfp{\Vsnot F}(s))
    \Compl\Vsnot F(t,\Fungfp{\Vsnot F}(E))\\
    &= \Union_{s\in D}\Vsnot F(F,\Fungfp{\Vsnot F}(s))
      \Compl\Union_{s\in D}\Vsnot F(s,\Fungfp{\Vsnot F}(E))
      \text{\quad by hom-cont.}\\
    &= \Union_{s\in D}\Vsnot F(s,\Fungfp{\Vsnot F}(s))
    = \Union_{s\in D}\Fungfp{\Vsnot F}(s)\,.
  \end{align*}
  In the second equation, we used the facts that $D$ is directed and
  the monotonicity of $\Vsnot F$ and $\Fungfp{\Vsnot F}$ on morphisms.

  Let $E\subseteq F$, we prove that
  $\Fungfp{\Vsnot F}(E)\subseteq\Fungfp{\Vsnot F}(F)$. This results
  from the observation that if $E'\subseteq F'$, then
  $\Vsnot F_E(E')\subseteq\Vsnot F_F(F')$ and hence
  $\forall n\in\Nat\ \Vsnot F_E^n(\emptyset)\subseteq\Vsnot
  F_F^n(\emptyset)$. Let us check that
  $\Fungfp{\Vsnot F}(\Relii_{E,F})=\Relii_{\Fungfp{\Vsnot
      F}(E),\Fungfp{\Vsnot F}(F)}\in\REL(\Fungfp{\Vsnot
    F}(E),\Fungfp{\Vsnot F}(F))$.
  We have
  \begin{multline*}
    \Vsnot F(F,\Fungfp{\Vsnot F}(\Relii_{E,F}))\Compl\Vsnot
    F(\Relii_{E,F},\Fungfp{\Vsnot F}(E))\\
    =\Vsnot F(\Relii_{E,F},\Fungfp{\Vsnot
      F}(\Relii_{E,F}))=\Fungfp{\Vsnot F}(\Relii_{E,F})
  \end{multline*} %
  by definition of the functor $\Fungfp{\Vsnot F}$ and
  \begin{align*}
    \Vsnot F(F,\Relii_{\Fungfp{\Vsnot F}(E),\Fungfp{\Vsnot F}(F)})\Compl\Vsnot
    F(\Relii_{E,F},\Fungfp{\Vsnot F}(E))
    &=\Relii_{\Vsnot F(E,\Fungfp{\Vsnot F}(E)),
      \Vsnot F(F,\Fungfp{\Vsnot F}(F))}\\
    &=\Relii_{\Fungfp{\Vsnot F}(E),\Fungfp{\Vsnot F}(F)}
  \end{align*}
  by strictness of $\Vsnot F$. The equation follows by
  Lemma~\ref{lemma:equations-final-coalgebra}, so that the functor
  $\Fungfp{\Vsnot F}$ is strict. 
\end{proof}

\subsection{Proof of Lemma~\ref{lemma:nuts-limpl-charact}}
\begin{proof}
  Let $t\in\Total{\Limpl XY}$ and let $u\in\Total X$. Let
  $v'\in\Orth{\Total Y}$, since
  $u\times v'\in\Total{\Tens X{\Orth Y}}$ we have
  $t\cap(u\times v')\not=\emptyset$ and hence
  $(\Matappa tu)\cap v'\not=\emptyset$. Therefore
  $\Matappa tu\in\Biorth{\Total Y}=\Total Y$. Conversely assume that
  $\forall u\in\Total X\ \Matappa tu\in\Total Y$. Let $u\in\Total X$
  and $v'\in\Total{\Orth Y}=\Orth{\Total Y}$. Since
  $\Matappa tu\in\Total Y$ we have
  $(\Matappa tu)\cap v'\not=\emptyset$ and hence
  $t\cap(u\times v')\not=\emptyset$ and this shows that
  $t\in\Total{\Limpl XY}$.
\end{proof}

\subsection{Proof of Lemma~\ref{lemma:NUTS-iso-charact}}
\begin{proof}
  Assume that $t$ is an iso in $\NUTS$ so that there is $t'\in\NUTS(Y,X)$
  such that $\Matapp{t'}t=\Id_{\Web X}$ and
  $\Matapp t{t'}=\Id_{\Web Y}$ and since we know that the isos in
  $\REL$ are the bijections we know that $t$ is a bijection. The fact
  that $\forall u\subseteq\Web X\ u\in\Total X\Equiv t(u)\in\Total Y$
  results from the fact that both $t$ and $t'=\Funinv t$ are morphisms
  in $\NUTS$.

  The converse implication is obvious.
\end{proof}

\subsection{Proof of Lemma~\ref{lemma:limpl-tens-charact}}
\begin{proof}
  The condition is obviously necessary, let us prove that it is
  sufficient so assume that $t$ fulfills it and let us prove that
  $t\in\Total{\Limpl{\Tens{X_1}{X_2}}{Y}}$. To this end it suffices to
  prove that
  $\Orth t\in\Total{\Limpl{\Orth Y}{\Orth{\Tensp{X_1}{X_2}}}}$. So let
  $v'\in\Total{\Orth Y}$ and let us prove that
  $\Matappa{\Orth
    t}{v'}\in\Total{\Orth{\Tensp{X_1}{X_2}}}=\Orth{\Eset{\Tens{u_1}{u_2}\St
      u_1\in\Total{X_1}\text{ and }u_2\in\Total{X_2}}}$. So let
  $u_i\in\Total{X_i}$ for $i=1,2$. We know that
  $\Matappa t{\Tensp{u_1}{u_2}}\in\Total Y$ and hence
  $\Matappap t{\Tensp{u_1}{u_2}}\cap v'\not=\emptyset$, that is
  $\Tensp{u_1}{u_2}\cap\Matappap{\Orth t}{v'}\not=\emptyset$, proving
  our contention.
\end{proof}

\subsection{Proof of Lemma~\ref{lemma:Assoc-tens-limpl}}
\begin{proof}
  Let $t\in\Total{\Limpl{\Tensp{X_1}{X_2}}{Y}}$ and let us prove that
  $s=\Matappa\Assoc t\in\Total{\Limpl{X_1}{\Limplp{X_2}{Y}}}$. Given
  $u_i\in\Total{X_i}$ is suffices to prove that
  $\Matappa{\Matappap{t'}{u_1}}{u_2}\in\Total Y$ which results from
  the fact that
  $\Matappa{\Matappap{s}{u_1}}{u_2}=\Matappa
  t{\Tensp{u_1}{u_2}}$. Conversely let
  $s\in\Total{\Limpl{X_1}{\Limplp{X_2}{Y}}}$ and let us prove that
  $t=\Matappa{\Funinv\Assoc}{s}\in\Total{\Limpl{\Tensp{X_1}{X_2}}{Y}}$. This
  results from lemma~\ref{lemma:limpl-tens-charact} and from the
  equation
  $\Matappa{\Matappap{s}{u_1}}{u_2}=\Matappa t{\Tensp{u_1}{u_2}}$.
\end{proof}

\subsection{Proof of Lemma~\ref{lemma:nuts-excl-map}}
\begin{proof}
  The condition is obviously necessary, so let us assume that it
  holds. By Lemma~\ref{lemma:nuts-orth-morph}, it suffices to prove
  that $\Orth t\in\NUTS(\Orth Y,\Orthp{\Excl X})$. Let
  $v'\in\Total{\Orth Y}$, we prove that
  $\Matappa{\Orth t}{v'}\in\Orth{\Total{\Excl Y}}$. So let
  $u\in\Total X$, since
  $\Matappa t{\Promhc u}\in\Total Y$ and hence
  $\Matappap t{\Promhc u}\cap v'\not=\emptyset$, that is
  $\Matappap{\Orth t}{v'}\cap\Promhc u\not=\emptyset$.
\end{proof}

\subsection{Proof of Lemma~\ref{lemma:nuts-excl-map-bil}}
\begin{proof}
  We deal with the case $k=2$. The condition is necessary since, if
  $u_1\in\Total{X_1}$ and $u_2\in\Total{X_2}$, then
  $\Tens{\Promhc{u_1}}{\Promhc{u_2}}\in\Total{\Tens{\Excl{X_1}}{\Excl{X_2}}}$.
  So assume that it holds. Let
  $t'=\Curlin(t)\in\REL(\Limpl{\Web{X_1}}{\Limplp{\Web{X_2}}{\Web
      Y}})$. Let $u_1\in\Total{X_1}$, we have
  $\Matappa{t'}{\Promhc{u_1}}\in\Part{\Web{\Limpl{\Excl{X_2}}Y}}$. Let
  $u_2\in\Total{X_2}$, we have
  $\Matappa{\Matappap{t'}{\Promhc{u_1}}}{\Promhc{u_1}} =\Matappa
  t{\Tensp{\Promhc{u_1}}{\Promhc{u_2}}}\in\Total Y$ by our
  assumption. It follows by Lemma~\ref{lemma:nuts-excl-map} that
  $\Matappa{t'}{\Promhc{u_1}}\in\Total{\Limpl{\Excl{X_2}}Y}$ and since
  this holds for any $u_1\in\Total{X_1}$ we actually have
  $t'\in\NUTS(\Excl{X_1},\Limpl{\Excl{X_2}}Y)$. It follows that
  $t=\Funinv{\Curlin}(t')\in\NUTS(\Tens{\Excl{X_1}}{\Excl{X_2}},Y)$ as
  contended.
\end{proof}

\subsection{Proof of Lemma~\ref{lemma:NUTS-excl}}
\begin{proof}
Given an object $X$ of $\NUTS$, we set
$\Der X=\Der{\Web X}\in\REL(\Web{\Excl X},\Web X)$ and
$\Digg X=\Digg{\Web X}\in\REL(\Web{\Excl X},\Web{\Excll X})$. Given
$u\in\Total X$, we have $\Matappa{\Der X}{\Promhc u}=u\in\Total X$ and
$\Matappa{\Digg X}{\Promhc u}=\Prommhc u\in\Total{\Excll X}$. It
follows by Lemma~\ref{lemma:nuts-excl-map} that
$\Der X\in\NUTS(\Excl X,X)$ and $\Digg X\in\NUTS(\Excl X,\Excll X)$.

Naturality and monadicity trivially hold because they hold in $\REL$:
we have an obvious faithful forgetful functor $\NUTS\to\REL$ which
commutes with all $\LL$ categorical constructs.

We are left with defining the strong monoidality structure of
$\Excl\_$ (Seely isomorphisms), for $\Seelyz\in\NUTS(\One,\Excl\Top)$
we take the same morphism as in $\REL$. And we set
$\Seelyt_{X_1,X_2}=\Seelyt_{\Web{X_1},\Web{X_2}}
\in\REL(\Web{\Tens{\Excl{X_1}}{\Excl{X_2}}},\Web{\Exclp{\With{X_1}{X_2}}})$. Let
$u_i\in\Total{X_i}$ for $i=1,2$. We have
\(
  \Matappa{\Seelyt_{X_1,X_2}}{\Tensp{\Promhc{u_1}}{\Promhc{u_2}}}
  =\Promhc{(\With{u_1}{u_2})}\in\Total{\Exclp{\With{X_1}{X_2}}}
  \)
since $\With{u_1}{u_2}\in\Total{\With{X_1}{X_2}}$, and hence by
Lemma~\ref{lemma:nuts-excl-map-bil} we have
\(
  \Seelyt_{X_1,X_2}\in\NUTS(\Tensp{\Excl{X_1}}{\Excl{X_2}},
  \Exclp{\With{X_1}{X_2}})
\).
Any element $w$ of $\Total{\With{X_1}{X_2}}$ is of shape
$w=\With{u_1}{u_2}$ with $u_i\in\Total{X_i}$, namely
$u_i=\Matappa{\Proj i}w$. We have
\(
  \Matappa{\Funinv{(\Seelyt_{X_1,X_2})}}{\Promhc w}
  =\Tens{\Promhc{u_1}}{\Promhc{u_2}}\in\Total{\Tens{\Excl{X_1}}{\Excl{X_2}}}
  \)
and hence by Lemma~\ref{lemma:nuts-excl-map} we have
$\Funinv{(\Seelyt_{X_1,X_2})}
\in\NUTS(\Exclp{\With{X_1}{X_2}},\Tensp{\Excl{X_1}}{\Excl{X_2}})$. This
ends the proof that $\NUTS$ is a model of classical Linear Logic since
the required commutations obviously hold because they hold in $\REL$.
\end{proof}


\subsection{Full proof of Theorem~\ref{th:VNUTS-model}}
\begin{proof}
Concerning Condition~(\ref{enum:seel-mull-3}), let
$(\Vsnot X_i)_{i=1}^k$ be elements of $\VNUTS n$ and let
$\Vsnot X\in\VNUTS k$. Considering $\Vsnot X$ and the $\Vsnot X_i$'s
as strong functors, we know that $\Vsnot X\Comp\Vect{\Vsnot X}$ is a
strong functor $\NUTS^n\to\NUTS$. We simply have to exhibit a VNUTS
whose associated strong functor is $\Vsnot X\Comp\Vect{\Vsnot X}$.
Let $\Vsnot F=\Web{\Vsnot X}\Comp\Web{\Vect{\Vsnot X}}$ (composition
of variable sets, Section~\ref{sec:strong-VS-Seely-model}). Let
$\Vect X\in\NUTS^n$, each $\Strfun{\Vsnot X_i}(\Vect X)$ is an object
of $\NUTS$ and hence\\
$(\Strfun{\Vsnot F}(\Web{\Vect{X}}), \Total{\Vsnot X}(\Strfun{\Vsnot
  X_1}(\Vect X),\dots,\Strfun{\Vsnot X_k}(\Vect X)))$ is a
NUTS. Moreover given $\Vect t\in\NUTS^n(\Vect X,\Vect Y)$, we know
that for each $i=1,\dots,k$, one has
$\Strfun{\Vsnot X_i}(\Vect t)\in\NUTS(\Strfun{\Vsnot X_i}(\Vect
X),\Strfun{\Vsnot X_i}(\Vect Y))$ since $\Vsnot X_i$ is a VNUTS. Since
$\Vsnot X$ is a VNUTS we have
\begin{multline*}
  \Strfun{\Vsnot F}(\Vect t)\\
  \in\NUTS(\Strfun{\Vsnot X}(\Strfun{\Vsnot
    X_1}(\Vect X),\dots,\Strfun{\Vsnot X_k}(\Vect X)),\Strfun{\Vsnot
    X}(\Strfun{\Vsnot X_1}(\Vect Y),\dots,\Strfun{\Vsnot X_k}(\Vect
  Y)))\,.
\end{multline*}

Let $X\in\Obj\NUTS$ and $\Vect Y\in\Obj{\NUTS^k}$. For $i=1,\dots,k$
we know that
$\Strnat{\Vsnot X_i}_{X,\Vect Y}\in\NUTS(\Tens{\Excl X}{\Strfun{\Vsnot
    X_i}(\Vect Y)},\Strfun{\Vsnot X_i}(\Tens{\Excl X}{\Vect
  Y}))$. Therefore
\begin{multline*}
  \Strfun{\Vsnot X}((\Strnat{\Vsnot X_i}_{X,\Vect
    Y})_{i=1}^k)\\
  \in\NUTS(\Strfun{\Vsnot
  X}((\Tens{\Excl X}{\Strfun{\Vsnot X_i}(\Vect Y)})_{i=1}^k),
\Strfun{\Vsnot X}((\Strfun{\Vsnot
  X_i}(\Tens{\Excl X}{\Vect Y}))_{i=1}^k))
\end{multline*}
and hence
\begin{multline*}
  \Strfun{\Vsnot X}((\Strnat{\Vsnot X_i}_{X,\Vect
    Y})_{i=1}^k)\Compl\Strnat{\Vsnot X}_{X,(\Strfun{\Vsnot X_i}(\Vect
    Y))_{i=1}^k}\\
  \in
  \NUTS(\Tens{\Excl X}{\Strfun{\Vsnot X}
((\Strfun{\Vsnot X_i}(\Vect Y))_{i=1}^k)},\Strfun{\Vsnot X}((\Strfun{\Vsnot
  X_i}(\Tens{\Excl X}{\Vect Y}))_{i=1}^k)
)\,.
\end{multline*}
Moreover we have
\begin{align*}
  \Strnat{\Vsnot F}_{\Web X,\Web{\Vect Y}}
  &=\Strfun{\Web{\Vsnot
    X}}((\Strnat{\Web{\Vsnot X_i}}_{\Web X,\Web{\Vect
    Y}})_{i=1}^k)\Compl\Strnat{\Web{\Vsnot X}}_{\Web X,(\Web{\Strfun{\Vsnot X_i}(\Vect
    Y)})_{i=1}^k}\\
  &\hspace{5em}\text{\quad by definition of }\Vsnot F\\
  &=\Strfun{\Web{\Vsnot
    X}}((\Strnat{\Web{\Vsnot X_i}}_{\Web X,\Web{\Vect
    Y}})_{i=1}^k)
    \Compl\Strnat{\Web{\Vsnot X}}_{\Web X,(\Strfun{\Web{\Vsnot X_i}}(\Web{\Vect
    Y}))_{i=1}^k}\\
  &=\Strfun{\Vsnot X}((\Strnat{\Vsnot X_i}_{X,\Vect
    Y})_{i=1}^k)\Compl\Strnat{\Vsnot X}_{X,(\Strfun{\Vsnot X_i}(\Vect
    Y))_{i=1}^k} 
\end{align*}
using again the fact that $\Vsnot X$ and the $\Vsnot{X}_i$'s are
VNUTS.
This shows that the pair $\Vsnot Y=(\Web{\Vsnot Y},\Total{\Vsnot Y})$
given by $\Web{\Vsnot Y}=\Vsnot F$ and
$\Total{\Vsnot Y}(\Vect X)=\Total{\Vsnot X}(\Strfun{\Vsnot X_1}(\Vect
X),\dots,\Strfun{\Vsnot X_k}(\Vect X))$ is a VNUTS whose associated
strong functor is $\Vsnot X\Comp\Vect{\Vsnot X}$ thus proving our contention.

Concerning Condition~(\ref{enum:seel-mull-4}), let us deal only with
the case of $\Excl\_$, the others being similar. We have to exhibit a
unary VNUTS $\Vsnot X$ whose associated strong functor $\NUTS\to\NUTS$
coincides with $\Excl\_$ (which is known to be a strong functor
$\NUTS\to\NUTS$ by Section~\ref{sec:NUTS-exponential} and by the
general considerations of
Section~\ref{sec:strong-fun-LL-operations}). For $\Web{\Vsnot X}$,
which has to be a variable set $\REL\to\REL$, we take the
interpretation $\Vsnot E$ of $\Excl\_$ in the model $\REL$
(Section~\ref{sec:strong-VS-Seely-model}) which is an element of
$\VREL 1$, that is, a unary variable set. Next, given
$X\in\Obj\NUTS$, we take $\Total{\Vsnot X}(X)=\Total{\Excl
  X}$. Condition~(\ref{enum:vnuts-cond-tot}) in the definition of
VNUTS holds by functoriality of $\Excl\_$ on
$\NUTS$. Condition~(\ref{enum:vnuts-cond-strength}) holds by the
definition of $\Strnat{\Vsnot F}_{\Web X,\Web Y}$ as described in
Section~\ref{sec:strong-fun-LL-operations} which coincides with
$\Monoidalt\Compl\Tensp{\Digg X}{\Excl Y}\in\NUTS(\Tens{\Excl X}{\Excl
  Y},\Excl{\Tensp{\Excl X}{Y}})$.

Let us now turn to Condition~(\ref{enum:seel-mull-5}) which is a bit more
challenging.

\subsubsection{Fixed Points of VNUTS}\label{sec:fix-VNUTS}
Let first $\Vsnot X=(\Web{\Vsnot X},\Total{\Vsnot X})$ be a \emph{unary}
VNUTS. Let $E=\Funfp{\Strfun{\Web{\Vsnot X}}}$ which is the least set
such that $\Strfun{\Web{\Vsnot X}}(E)=E$, that is
$E=\Union_{n=0}^\infty\Strfun{\Web{\Vsnot X}}^n(\emptyset)$. Let
$\Phi:\Tot E\to\Tot E$ be defined as follows: given $\cS\in\Tot E$,
then $(E,\cS)$ is a NUTS, and we set
$\Phi(\cS)=\Total{\Vsnot X}(E,\cS)\in\Tot{\Strfun{\Web{\Vsnot
      X}}(E)}=\Tot E$.
This function $\Phi$ is monotone. Let indeed $\cS_1,\cS_2\in\Tot E$
with $\cS_1\subseteq \cS_2$. Then we have
$\Id\in\NUTS((E,\cS_1),(E,\cS_2))$ and therefore, by
Condition~(\ref{enum:vnuts-cond-tot}) satisfied by $\Vsnot X$, we have
\begin{align*}
  \Id=\Strfun{\Web{\Vsnot X}}(\Id)
  &\in\NUTS(\Strfun{\Vsnot
    X}(E,\cS_1),\Strfun{\Vsnot
    X}(E,\cS_2))\\
  &\hspace{3em}=\NUTS((E,\Phi(\cS_1)),(E,\Phi(\cS_2))
\end{align*}
which means that $\Phi(\cS_1)\subseteq\Phi(\cS_2)$.
By the Knaster Tarski Theorem (remember that $\Tot E$ is a complete
lattice), $\Phi$ has a greatest fixpoint $\cT$ that we can describe as
follows. Let $(\cT_\alpha)_{\alpha\in\Ordinals}$, where $\Ordinals$ is
the class of ordinals, be defined by: $\cT_0=\Part E$ (the largest
possible notion of totality on $E$), $\cT_{\alpha+1}=\Phi(\cT_\alpha)$
and $\cT_\lambda=\Inter_{\alpha<\lambda}\cT_\alpha$ when $\lambda$ is
a limit ordinal. This sequence is decreasing (easy induction on
ordinals using the monotonicity of $\Phi$) and there is an ordinal
$\theta$ such that $\cT_{\theta+1}=\cT_\theta$ (by a cardinality
argument; we can assume that $\theta$ is the least such ordinal). The
greatest fixpoint of $\Phi$ is then $\cT_\theta$ as easily checked.

By construction $((E,\cT_\theta),\Id)$ is an object of
$\COALGFUN{\NUTS}{\Strfun{\Vsnot X}}$, we prove that it is the
final object. So let $(Y,t)$ be another object of the same
category. Since $(\Web Y,t)$ is an object of
$\COALGFUN\REL{\Strfun{\Web{\Vsnot X}}}$ and since $(E,\Id)$ is the
final object in that category, we know by
Lemma~\ref{lemma:rel-fixpoint-final} that there is exactly one
$e\in\REL(\Web Y,E)$ such that $\Strfun{\Web{\Vsnot X}}(e)\Compl
t=e$. We prove that actually $e\in\NUTS(Y,(E,\cT_\theta))$ so let
$v\in\Total Y$. We prove by induction on the ordinal $\alpha$ that
$\Matappa ev\in\cT_\alpha$. For $\alpha=0$ it is obvious since
$\cT_0=\Part E$. Assume that the property holds for $\alpha$ and let
us prove it for $\alpha+1$. We have
$\Matappa tv\in\Total{\Vsnot X}(Y)=\Total{\Strfun{\Vsnot X}(Y)}$ since
$t\in\NUTS(Y,\Strfun{\Vsnot X}(Y))$. Since
$\Strfun{\Vsnot X}(e)\in\NUTS(\Strfun{\Vsnot X}(Y),\Strfun{\Vsnot
  X}(E,\cT_\alpha))$ and since
$\Strfun{\Vsnot X}(E,\cT_\alpha)=(E,\cT_{\alpha+1})$ we have
$\Matappa{(\Strfun{\Vsnot X}(e)\Compl t)}v\in\cT_{\alpha+1}$, that is
$\Matappa ev\in\cT_{\alpha+1}$. Last if $\lambda$ is a limit ordinal
and if we assume $\forall\alpha<\lambda\ \Matappa ev\in\cT_\alpha$ we
have $\Matappa
ev\in\Inter_{\alpha<\lambda}\cT_\alpha=\cT_\lambda$. Therefore
$\Matappa ev\in\cT_\theta$. We use $\Fungfp{\Strfun{\Vsnot X}}$ to
denote this final coalgebra $(E,\cT_\theta)$ (its definition depends
only on $\Strfun{\Vsnot X}$ and does not involve the strength
$\Strnat{\Vsnot X}$).

So we have proven the first part of Condition~(\ref{enum:seel-mull-5})
in the definition of a Seely model of $\MULL$ (see
Section~\ref{def:categorical-muLL-models}). As to the second part, let
$\Vsnot X$ be an $n+1$-ary VNUTS. We know by the general
Lemma~\ref{lemma:strfun-gfp-general} that there is a uniquely defined
strong functor $\Fungfp{\Vsnot X}:\NUTS^n\to\NUTS$ such that
\begin{Itemize}
\item
  $\Strfun{\Fungfp{\Vsnot X}}(\Vect X)=\Fungfp{(\Strfun{\Vsnot
      X}_{\Vect X})}$, so that
  $\Strfun{\Vsnot X}(\Vect X,\Strfun{\Fungfp{\Vsnot X}}(\Vect
  X))=\Strfun{\Fungfp{\Vsnot X}}(\Vect X)$, for all
  $\Vect X\in\Obj{\NUTS^n}$,
\item
  $\Strfun{\Vsnot X}(\Vect t,\Strfun{\Fungfp{\Vsnot X}}(\Vect
  t))=\Strfun{\Fungfp{\Vsnot X}}(\Vect t)$ for all
  $\Vect t\in\NUTS(\Vect X,\Vect Y)$
\item and
  $\Strfun{\Vsnot X}(\Tens{Y}{\Vect X},\Strnat{\Fungfp{\Vsnot
      X}}_{Y,\Vect X})\Compl\Strnat{\Vsnot X}_{Y,(\Vect
    X,\Strfun{\Fungfp{\Vsnot X}}(\Vect X))}=\Strnat{\Fungfp{\Vsnot X}}_{Y,\Vect
    X}$ for all $Y\in\Obj\NUTS$ and $\Vect X\in\Obj{\NUTS^n}$.
\end{Itemize}
To end the proof, it will be enough to exhibit an $n$-ary VNUTS
$\Vsnot Y=(\Web{\Vsnot Y},\Total{\Vsnot Y})$ whose associated strong
functor coincides with $\Fungfp{\Vsnot X}$. We know that
$\Web{\Vsnot X}$ is a variable set $\REL^{n+1}\to\REL$ so let
$\Vsnot F=\Fungfp{\Web{\Vsnot X}}=\Funfp{\Web{\Vsnot X}}$ which is a
variable set $\REL^n\to\REL$ (see Section~\ref{sec:VS-fixpoints}). Let
$\Vect X\in\Obj{\NUTS^n}$, we have
$\Web{\Strfun{\Fungfp{\Vsnot X}}(\Vect
  X)}=\Web{\Fungfp{(\Strfun{\Vsnot X}_{\Vect
      X})}}=\Union_{n=0}^\infty\Web{\Strfun{\Vsnot X}_{\Vect
    X}}^n(\emptyset)=\Strfun{\Vsnot F}(\Web{\Vect X})$. Let
$\Vect t\in\NUTS^n(\Vect X,\Vect Y)$, then
$\Strfun{\Fungfp{\Vsnot X}}(\Vect t)$ is the unique element $s$ of
$\NUTS(\Strfun{\Fungfp{\Vsnot X}}(\Vect X),\Strfun{\Fungfp{\Vsnot
    X}}(\Vect Y))$ (this hom-set is a subset of
$\REL(\Vsnot F(\Web{\Vect X}),\Vsnot F(\Web{\Vect Y}))$) which
satisfies $\Strfun{\Vsnot X}(\Vect t,s)=s$, that is
$\Strfun{\Web{\Vsnot X}}(\Vect t,s)=s$. This means that
$\Strfun{\Fungfp{\Vsnot X}}(\Vect t)=s=\Strfun{\Vsnot F}(\Vect t)$. By
a completely similar uniqueness argument we have
$\Strnat{\Fungfp{\Vsnot X}}_{X,\Vect Y}=\Strnat{\Vsnot F}_{\Web
  X,\Web{\Vect Y}}$ for all $X\in\Obj{\NUTS}$ and
$\Vect Y\in\Obj{\NUTS^n}$. So we set $\Web{\Vsnot Y}=\Vsnot F$.

Next, given $\Vect X\in\Obj{\NUTS^n}$ we set
\begin{multline*}
  \Total{\Vsnot Y}(\Vect X)=\Total{\Strfun{\Fungfp{\Vsnot X}}(\Vect
    X)}\\
  \in\Tot{\Web{\Strfun{\Fungfp{\Vsnot X}}(\Vect
  X)}}=\Tot{\Strfun{\Vsnot F}(\Web{\Vect X})}
\end{multline*}
Given
$\Vect t\in\NUTS(\Vect X,\Vect Y)$ we have
\begin{multline*}
  \Strfun{\Vsnot F}(\Vect t)=\Strfun{\Fungfp{\Vsnot X}}(\Vect
  t)\\
  \in\NUTS((\Strfun{\Vsnot F}(\Web{\Vect X}),\Total{\Vsnot Y}(\Vect
  X)),(\Strfun{\Vsnot F}(\Web{\Vect Y}),\Total{\Vsnot Y}(\Vect Y))
\end{multline*}
since
$(\Strfun{\Vsnot F}(\Web{\Vect X}),\Total{\Vsnot Y}(\Vect
X))=\Strfun{\Fungfp{\Vsnot X}}(\Vect X)$ and similarly for $\Vect
Y$. Last since
$\Strnat{\Vsnot F}_{\Web X,\Web{\Vect Y}}=\Strnat{\Fungfp{\Vsnot
    X}}_{X,\Vect Y}\in\NUTS(\Tens{\Excl X}{\Strfun{\Fungfp{\Vsnot
      X}}(\Vect Y)},\Strfun{\Fungfp{\Vsnot X}}(\Tens X{\Vect Y}))$ we
know that $\Vsnot Y=(\Web{\Vsnot Y},\Total{\Vsnot Y})$ is a VNUTS
whose associated strong functor is $\Fungfp{\Vsnot X}$. This ends the
proof that $(\NUTS,(\VNUTS n)_{n\in\Nat})$ is a Seely model of
$\MULL$.
\end{proof}






\end{document}